\documentclass{article}
\usepackage{mathrsfs}
\usepackage{a4wide}
\usepackage{amsmath,amsfonts,amssymb,amsthm,color}
\usepackage[usenames,dvipsnames,svgnames,table]{xcolor}

\usepackage[active]{srcltx}

\allowdisplaybreaks

\def \beq{\begin{equation}}
\def \eeq{\end{equation}}

\newcommand{\R}{{\mathbb R}}

\newcommand{\X}{{\mathcal{X}}}
\newcommand{\z}{\text{\tt{z}}}
\def\bb1{{1\hspace*{-2.4pt}\rm{l}}}
\newcommand {\curl} {{\rm curl}\,}

\newtheorem{theorem}{Theorem}[section]
\newtheorem{definition}[theorem]{Definition}
\newtheorem{proposition}[theorem]{Proposition}
\newtheorem{corollary}[theorem]{Corollary}
\newtheorem{lemma}[theorem]{Lemma}

\theoremstyle{definition}
\newtheorem{remark}[theorem]{Remark}

\newtheorem{hypothesis}[theorem]{Hypothesis}

\numberwithin{equation}{section}

\def\R{{\rm I\kern-.2em R}}

\def\X{\mathcal X}

\def\s0{\sigma_0}

\def\z{\mathfrak{z}}

\def\bb1{{\rm{1}\hspace{-3pt}\mathbf{l}}}
\def\Ie0{[-\epsilon_0,\epsilon_0]}

\def\Int{\mathfrak{I}\mathit{nt}}
\def\lambdabar{\lambda\hspace{-9pt}\relbar}

\def\beq{\begin{equation}}
\def\eeq{\end{equation}}

\begin{document}

\noindent 
\begin{center}
\textbf{\large Peierls' substitution for low lying spectral energy windows }
\vspace{0,3cm}
\end{center}

\vspace{0.2cm}

\begin{center}
{\bf
Horia D. Cornean\footnote{Department of Mathematical Sciences,
  Aalborg University, Skjernvej 4A, 9220 Aalborg, Denmark}, \ 
  Bernard Helffer\footnote{Laboratoire
de Math\'ematiques Jean Leray, Universit\'{e} de Nantes and CNRS, 2 rue de la Houssini\`ere 44322 Nantes Cedex (France) and Laboratoire de Math\'ematiques d'Orsay, Univ.
Paris-Sud, Universit\'e Paris-Saclay.},  and Radu
Purice\footnote{Institute
of Mathematics Simion Stoilow of the Romanian Academy, P.O.  Box
1-764, Bucharest, RO-70700, Romania.}}
\end{center}

\abstract{

We consider a $2d$ magnetic Schr\"odinger operator perturbed by a weak magnetic field which slowly varies around a positive mean.  In a previous paper we proved the  appearance of a `Landau
type' structure of spectral islands at the bottom of the spectrum, under the
hypothesis that the lowest Bloch eigenvalue of the unperturbed operator remained simple on the whole Brillouin zone, even though its range may overlap with the range of the second eigenvalue.  We also assumed that the first Bloch spectral projection was smooth and had a  zero Chern number.

In this paper we extend our previous results to the only two remaining  possibilities: either the first Bloch eigenvalue remains isolated while its corresponding spectral projection has a non-zero Chern number, or the first two Bloch eigenvalues cross each other.

\tableofcontents

\section{Introduction}\label{introduc}

\subsection{Some history of the problem and the structure of the paper}
Since the pioneering work of Peierls \cite{Pe} and Luttinger \cite{Lu}, physicists consider Peierls' substitution to be the right way of taking into account external magnetic fields in periodic solid state systems. Very roughly, the receipt is as follows: consider a non-magnetic periodic Schr\"odinger Hamiltonian  having an energy band $\lambda_0(\theta)$ which is periodic with respect to the Brillouin zone; then in the presence of an external magnetic vector potential $A$, the right physical behaviour of the magnetically perturbed system should be given by $\lambda_0(\theta -A(x))$, whatever that means. 

From a rigorous mathematical point of view the situation is more complicated, and here we will only discuss the case in which the external magnetic perturbation is weak and slowly variable. The smallness of the magnetic perturbation is generically represented by a parameter $\epsilon\geq 0$, where $\epsilon=0$ corresponds to the unperturbed objects. The first natural question one asks in this context is whether one can construct an effective Hamiltonian which is close in some sense to the original one, either dynamically or spectrally. The second question is to analyse the effective model and draw conclusions on its spectral and dynamical properties. 

To the best of our knowledge, all the effective models constructed until now suppose that the unperturbed Bloch
spectrum under investigation remains isolated from the rest; this does not require a gap in the full spectrum, only a non-empty gap at each fixed quasi-momentum $\theta$ of the Brillouin zone (the non-crossing, range-overlapping case). On top of that, with the notable exception of a recent work by Freund and Teufel \cite{FT}, all existing papers require that the Bloch projection associated to the Bloch spectrum under investigation is {\it topologically trivial}. The triviality condition leads to the existence of a basis consisting of spatially localized Wannier functions \cite{CHN, CM, FMP, HS4, Ne-LMP, Ne-RMP} which plays an essential role in the reduction procedure. We note that although \cite{FT} can treat non-trivial projections, only bounded perturbations are allowed. 

When the external magnetic field is constant and the range of $\lambda_0$ does not overlap with the other ones (the gapped case), one can construct a smooth and periodic symbol $\lambda_\epsilon(\xi)$ such that the Weyl quantization of $\lambda_\epsilon(\xi -A_\epsilon(x))$ is {\it isospectral} with the "true" magnetic band Hamiltonian, provided $\epsilon$ is small enough. This is essentially due to Nenciu and Helffer-Sj\"ostrand \cite{HS4, Ne-RMP, Sj}. Concerning the spectral analysis of these effective magnetic pseudo-differential operators there exists a large amount of literature which classifies the  spectral gaps induced by the magnetic field; see the works by Bellissard \cite{Be1, Be2}, Bellissard-Rammal \cite{RB}, Helffer-Sj\"ostrand \cite{HS, HS4} and references therein.

When the magnetic perturbation comes from a non-constant but slowly variable magnetic field, still under the assumption that a localized Wannier basis exists, one can again construct an effective Hamiltonian up to any order in $\epsilon$ (see \cite{CIP, DN-L, PST}) but which typically lives in an $\epsilon$-dependent subspace. Thus the spectral analysis of the effective operator seems to be as complicated as the original one. 

In our previous paper \cite{CHP} we proved (for the first time in the non-constant case) the appearance of a `Landau
type' structure at the bottom of the spectrum consisting of spectral islands of width of order $\epsilon$, separated by gaps of roughly the same order of magnitude; in \cite{CHP} we still worked under the hypothesis that the lowest Bloch eigenvalue of the unperturbed operator remained simple on the whole Brillouin zone and a localized Wannier basis can be constructed.

The current manuscript extends our previous results to the only two remaining possibilities: either the first Bloch eigenvalue remains isolated while its corresponding spectral projection has a non-zero Chern number, or the first two Bloch eigenvalues cross each other. 

Our main result is Theorem \ref{mainTh}. It roughly states that in a certain narrow energy window near the bottom of the spectrum of the "true" perturbed Hamiltonian, the perturbed spectrum is well approximated by the spectrum of a pseudo-differential operator corresponding to a magnetically quantized periodic symbol which only depends on $\xi$. This result is new also for the constant magnetic case since we no longer need global non-crossing and/or triviality conditions on the lowest Bloch eigenvalue. 

Another important consequence is that we can prove that gaps of order $\epsilon$ open near the bottom of the perturbed spectrum, due to the fact that the spectrum of our effective Hamiltonian has the same property. For more details see Corollaries \ref{C1} and \ref{C2}. 

The paper is organized as follows:  
\begin{itemize}
\item In the rest of this section we introduce some of the objects we want to study and we briefly review the Bloch-Floquet theory.
\item In Section \ref{sectiunea2} we state our main results (Theorem \ref{mainTh} and Corollaries \ref{C1} and \ref{C2}). In subsection \ref{ss2.2} we outline the main ideas behind the proofs. 
\item In Section \ref{S_q-W} we construct a Wannier-like localized basis which spans a certain energy window of the unperturbed Hamiltonian; in this window, the lowest Bloch eigenvalue has to be non-degenerate.  
\item In Section \ref{S.4} we describe a general reduction argument based on the Feshbach-Schur map which permits the construction of an effective Hamiltonian in a narrow window near the bottom of the spectrum.
\item In Section \ref{ss1.3} we review the so-called magnetic pseudo-differential calculus where the symbol composition rule is a twisted Moyal product, while in Section \ref{S.5} we construct the magnetic counterpart of the Wannier-like projection from Section \ref{S_q-W}.
\item In Section \ref{S.6} we show that the magnetic Hamiltonian satisfies the conditions of the abstract reduction procedure which we developed in Section \ref{S.4}. 
\item In Section \ref{modif-B-function} we analyse the structure of the effective magnetic matrix constructed in Section \ref{S.6}. In particular, we show in Proposition \ref{magn-matrix-ke} that the spectrum of this magnetic matrix is close to the spectrum of a magnetic pseudo-differential operator whose symbol $\lambdabar^\epsilon(\xi)$ is smooth and periodic in $\xi$ while  independent of $x$. In particular, this completes the proof of Theorem \ref{mainTh}(i). In Proposition \ref{mod-Bloch-func} we show that the function $\lambdabar^\epsilon(\cdot)$ is close to the unperturbed Bloch eigenvalue $\lambda_0$ near its minimum; this ends the proof of Theorem \ref{mainTh}. Finally, we prove Corollaries \ref{C1} and \ref{C2}. 
\end{itemize}

\subsection{Notation and conventions}

For any real finite dimensional Euclidean space $\mathcal{V}$ we shall denote by:
\begin{enumerate}
\item $BC^\infty(\mathcal{V})$ (respectively $C^\infty_{\text{\tt pol}}(\mathcal{V})$) the family of smooth complex functions on $\mathcal{V}$ that are bounded (respectively polynomially bounded) together with all their derivatives,
\item $\tau_v$ the translation by $-v\in\mathcal{V}$ on any class of distributions on $\mathcal{V}$, 
\item the "Japanese bracket" is denoted by $<v>:=\sqrt{1+|v|^2}$ for any $v\in\mathcal{V}$ with $|v|\in\mathbb{R}_+$ its Euclidean norm, 
\item $\mathscr{S}(\mathcal{V})$ the space of Schwartz  functions on $\mathcal{V}$ 
 and with $\mathscr{S}^\prime(\mathcal{V})$ the space of tempered distributions on $\mathcal{V}$.
\end{enumerate}
For any Banach space $\mathcal{B}$ we denote by $\mathcal{L}(\mathcal{B})$ the Banach space of continuous linear operators in $\mathcal{B}$. For any Hilbert space $\mathcal{H}$ and any pair of unit vectors $(u,v)$ in it we use the physics notation  $|u\rangle\langle v|$ for the 
 the projector
\beq
\mathcal H \ni w \mapsto |u\rangle\langle v|w:=\langle v,w\rangle_{\mathcal{H}}u\,.
\eeq
Our scalar products are anti-linear in the first factor. 
As we are working in a two-dimensional framework, we  use the following
notation for the {\it vector product} of two vectors $u$ and
$v$ in $\mathbb R^2$\,:
\beq\label{Def-wedge}
u\wedge v\ :=\ u_1v_2\,-\,u_2v_1\,.
\eeq
and the $-\pi/2$ rotation of $v$:
\beq\label{Def-bot}
v^\bot:=\big(v_2,-v_1\big)\,.
\eeq

\subsection{The Hamiltonian}

Let $\Gamma\subset\mathbb{R}^2\equiv\mathcal{X}$ be a regular periodic
lattice. We consider a smooth $\Gamma$-periodic
potential $V_\Gamma:\X\rightarrow\mathbb{R}$ and a smooth $\Gamma$-periodic magnetic field $B^\Gamma:\X\rightarrow\mathbb{R}$ with vanishing flux through the unit cell of the lattice. 
Let us denote by $A^\Gamma:\X\rightarrow\mathbb{R}^2$ a smooth $\Gamma$-periodic vector potential generating the magnetic field $B^\Gamma$. 
We  consider the Hilbert space
$
\mathcal{H}\,:=\,L^2(\X)
$
and define the periodic Schr\"{o}dinger operator
\beq\label{H0}
H^0:=(-i\partial_{x_1} - A^{\Gamma}_1)^2 +
(-i\partial_{x_2} - A^{\Gamma}_2)^2 + V_\Gamma\,,
\eeq
as the unique self-adjoint extension in $\mathcal{H}$ of the above differential operator initially defined on $\mathscr{S}(\mathcal{V})$ (see \cite{RS-4}).

As a perturbation of the periodic case, we consider a family of magnetic
fields indexed by $(\epsilon,\kappa)\in[0,1]\times[0,1]$:
\beq\label{Bek}
B^{\epsilon,\kappa}(x)\,:=\,\epsilon B_0\,+\,\kappa\epsilon B(\epsilon x)\,.
\eeq
Here $B_0>0$ is constant, while  
$B:\X\rightarrow\mathbb{R}$ is of class $BC^\infty(\X)$.  
We choose the vector potential $A^0$ of the constant magnetic field $B_0$ in   the form
{\it (transverse gauge)}:
\beq \label{A0}
A^0(x)\,=\,(B_0/2)\big(-x_2,x_1\big)
\eeq
and some {\it vector potential} 
$A:\X\rightarrow\mathbb{R}^2$, that we can choose to be of class $C^\infty_{\text{\tt pol}}(\X)$ such that:
\begin{equation}\label{B0}
B_0=\partial_1A^0_2-\partial_2A^0_1\,,\quad B=\partial_1A_2-\partial_2A_1\,.
\end{equation}
Similarly, for 
$$
A^{\epsilon,\kappa}(x):=\epsilon A^0(x)+\kappa A(\epsilon x)\,,
$$
we have 
\begin{equation}\label{Aek}
B^{\epsilon,\kappa}=\partial_1A^{\epsilon,\kappa}_2-\partial_2A^{\epsilon,
\kappa}_1\,.
\end{equation}
We shall also use the notation:
\beq\label{B}
B^\epsilon:=\epsilon B_0\,,\quad A^\epsilon:=\epsilon
A_0,\quad B_\epsilon(x):=B(\epsilon x)\,,\quad A_\epsilon(x):=A(\epsilon x)\,.
\eeq

\begin{definition}
We consider the  magnetic Schr\"{o}dinger operators defined as the self-adjoint extension in $\mathcal{H}$ of the following differential operators:
\begin{equation}\label{mainH}
H^{\epsilon,\kappa}:=(-i\partial_{x_1}-A^\Gamma_1 - A^{\epsilon,\kappa}_1)^2 +
(-i\partial_{x_2} -A^\Gamma_2- A^{\epsilon,\kappa}_2)^2 + V_\Gamma\,.
\end{equation} 
\begin{equation}\label{constmH}
H^{\epsilon}:=H^{\epsilon,0}=(-i\partial_{x_1} -A^\Gamma_1- A^{\epsilon}_1)^2 +
(-i\partial_{x_2} -A^\Gamma_2- A^{\epsilon}_2)^2 + V_\Gamma\,.
\end{equation} 
\end{definition}
When $\kappa=\epsilon=0$, we recover $H^{0}$.

\subsection{The Bloch-Floquet Theory}\label{BF-theory}

We  consider
the quotient space $\R^2/\Gamma$ which is canonically
isomorphic to the two dimensional torus $\mathbb{T}^2\equiv \mathbb{T}$
and denote by $\R^2\ni
x\mapsto\hat{x}\in\mathbb{T}$ the canonical
projection onto the quotient. We choose two generators $\{e_1,e_2\}$ for the lattice $\Gamma$ and the associate 
\textit{elementary cell}:
\beq\label{E-cell}
E_\Gamma\,=\,\left\{y=\sum\limits_{j=1}^2 t_je_j\in\R^2\,\mid\,-(1/2)\leq
t_j<(1/2)\,,\ \forall j\in\{1,2\}\right\},
\eeq
so that we have a bijection
$
\X\ni x\mapsto\big(\gamma(x),y(x)\big)\in\Gamma\times E_\Gamma\, .$ 
The dual lattice of $\Gamma$ is then defined as
$$
\Gamma_*\,:=\,\left\{\gamma^*\in\X^*\,\mid\,<\gamma^*,\gamma>\in\,2\pi\,\mathbb{Z}\,
,\,
\forall\gamma\in\Gamma\right\}.
$$
Considering the dual basis
$\{e^*_1,e^*_2\}\subset\X^*$ of
$\{e_1,e_2\}$, which is defined by
$<e^*_j,e_k>=\,2\pi\, \delta_{jk}\,,$ 
 we have  $\Gamma_*:=\oplus_{j=1,2}\, \mathbb{Z}e^*_j\,.
 $
 We
 similarly define $\mathbb{T}_*:=\X^*/\Gamma_*$  and $
E_{\Gamma_*}\equiv E_*$.

We recall that the Floquet unitary map $\mathcal{U}_{\Gamma}$ is defined by:
\begin{equation}\label{UGamma}
\mathscr{S}(\X)\ni\phi\mapsto\mathcal{U}_{\Gamma}\phi\in
C^\infty\big(\X\times\mathbb{T}_*\big) \mbox{ with }\quad\big[\mathcal{U}_{\Gamma}\phi\big]
(x,\theta):=\sum\limits_{\gamma\in\Gamma}e^{i<\theta,\gamma>}\phi(x-\gamma)\,,
\end{equation}
so that we have the following property
\beq\label{F-cond}
\big[\mathcal{U}_{\Gamma}\phi\big](x+\gamma,\theta)=e^{i<\theta,\gamma>}\big[
\mathcal{U}_{\Gamma}\phi\big](x,\theta)\,,\,
\forall\gamma\in\Gamma\,,\ \forall(x,\theta)\in\X\times\mathbb{T}_*\,.
\eeq
For any $\theta\in\mathbb{T}_*$ we introduce
\begin{equation}\label{F-theta}
\mathscr{F}_\theta\,:=\,\big\{v\in L^2_{\text{\sf
loc}}(\X)\,\mid\,v(x+\gamma)=e^{i<\theta,\gamma>}v(x)\,,\,
\forall\gamma\in\Gamma\big\}
\end{equation}
and notice that it is a Hilbert space for the scalar product
\begin{equation}\label{F-norm}
\langle v\,,\, w\rangle_{\mathscr{F}_\theta}\,:= \int_{E}\bar v(x)\, w (x)\,dx\,.
\end{equation}
For each $\theta\in\mathbb{T}_*$, we define  
$\mathcal{V}_\theta:\mathscr{F}_\theta\rightarrow L^2(\mathbb{T})$ as the multiplication  with the function \break 
$\sigma_\theta(x):=e^{-i<\theta,x>}$. We can then form the associated direct integral of Hilbert spaces
\beq\label{F}
\mathscr{F}:=\int_{\mathbb{T}_*}^\oplus\mathscr{F}_\theta\,d\theta\,,
\eeq
and the following unitary map from $ \int_{\mathbb{T}_*}^
\oplus\mathscr{F}_\theta\,d\theta$ onto
$L^2(E_*)\otimes L^2(\mathbb{T})$ defined by:
\beq\label{V}
\mathcal{V}_{\mathbb{T}_*}:=\int_{\mathbb{T}_*}^
\oplus\mathcal{V}_\theta\,d\theta\,.
\eeq
Then $\mathcal{U}_{\Gamma}$ defines a unitary
operator
$L^2\big(\X\big)\rightarrow\mathscr{F}$. 

Note that 
a bounded operator $X\in\mathcal{L}(\mathcal{H})$ commutes with all the translations with elements from $\Gamma$ if and only if there exists a measurable family $\mathbb{T}_*\ni\theta\mapsto\hat{X}(\theta)\in\mathcal{L}(\mathscr{F}_\theta)$ such that
$$
\mathscr{U}_\Gamma X\mathscr{U}_\Gamma^{-1}\ =\ \int^\oplus_{\mathbb{T}_*}\hat{X}(\theta)\,d\theta\,.
$$

The following statements are well-known 
(see for example \cite{Ku, RS-4}):
\begin{enumerate}
\item We have the direct integral decomposition: 
\begin{equation}\label{fibop}
\hat{H}^0:=\mathcal{U}_\Gamma H^0\mathcal{U}_\Gamma^{-1}=
\int_{\mathbb{T}_*}^\oplus\hat{H}^0(\theta)\, d\theta\,,
\end{equation}
with 
$$
{\hat{H}^0(\theta):=(-i\partial_{x_1} - A^{\Gamma}_1)^2 +
(-i\partial_{x_2} - A^{\Gamma}_2)^2 + V_\Gamma}\,,
 $$
 whose domain in $ \mathscr{F}_\theta$ is the following local Sobolev space:
 \beq
 \mathscr{H}^2_\theta\ :=\ \,\big\{v\in\mathcal{H}^2_{\text{\sf
loc}}(\X)\,\mid\,v(x+\gamma)=e^{i<\theta,\gamma>}v(x)\,,\,
\forall\gamma\in\Gamma\big\}\,\subset\,\mathscr{F}_\theta\,.
 \eeq
\item  The family $\hat{H}^0(\theta)$ ($\theta \in \mathbb{T}_*$)  is unitarily equivalent with an 
analytic family of type A in the sense of Kato \cite{Ka}. The unitary operator is given by the multiplication with $\mathscr{V}_\theta$ which maps $\mathcal{F}_\theta$ to $L^2(\mathbb{T})$. The rotated operator is essentially self-adjoint on the set of smooth and periodic functions in $L^2(E)$ and acts on them as:  
$$\widetilde{H}^0(\theta):=(-i\nabla - A^{\Gamma}+\theta)^2  + V_\Gamma,\quad \theta\in E_*\,.$$
\item
There exists a family of continuous functions
$\mathbb{T}_*\ni\theta\mapsto\lambda_j(\theta)\in
\mathbb{R}$  indexed by
$j\in\mathbb{N}$, called \textit{the Bloch eigenvalues},
such that $\lambda_j(\theta)\leq\lambda_{j+1}(\theta)$ for every
$j\in\mathbb{N}$ and $\theta\in \mathbb{T}_*$, and
\beq\label{Bloch}
\sigma\big(\hat{H}^0(\theta)\big)=\underset{j\in\mathbb{N}}{\bigcup}\{
\lambda_j(\theta)\}=\sigma\big(\widetilde{H}^0(\theta
)\big).
\eeq

\item 
For each fixed $\theta$ we can define the Riesz spectral projections using complex Cauchy integrals of the resolvent $(z- \hat{H}^0(\theta))^{-1}$ around each eigenvalue $\lambda_j(\theta)$. If the lowest eigenvalue $\lambda_0(\theta)$ is always non-degenerate, it is also smooth in $\theta$ and its corresponding Riesz projection $\hat{
\pi}_0(\theta)$ is globally smooth on the dual torus and has rank one. If $\lambda_0$ can cross with the other higher eigenvalues, then by convention, at these crossing points $\hat{
\pi}_0(\theta)$ is defined as the higher-rank Riesz projection which encircles the still lowest, but now degenerate eigenvalue. In this way,  the projection $\hat{
\pi}_0(\theta)$ is well defined on the whole dual torus but it is no longer continuous.
\end{enumerate}

\section{The main results}\label{sectiunea2}

\subsection{(Non)existence of localized Wannier functions}

In our previous work \cite{CHP} we imposed the condition that $\lambda_0$ had to remain non-degenerate on the whole dual torus, while the unperturbed periodic operator $H^0$ had to be time-reversal invariant (i.e. commuting with the complex conjugation operator). Under these conditions, the orthogonal projection family $\hat{
\pi}_0(\theta)$ is smooth and time-reversal symmetric in the sense that 
$$ C \, \hat{
\pi}_0(\theta)=\hat{
\pi}_0(-\theta) \, C $$
where $C$ is the anti-unitary operator given by the complex conjugation. In particular, this implies that the Chern number of this family equals zero and that one can construct \cite{HS4, Ne-RMP, FMP, CM} a global smooth  section $\hat{\phi}_0:\mathbb{T}_*\rightarrow\mathscr{F}$ such that
\begin{equation}\label{defhatpi} 
 \hat{
\pi}_0(\theta)=|\hat{\phi}_0(\theta)\rangle\langle\hat{\phi}_0(\theta)|\,.
\end{equation}
Thus we have the orthogonal projection in $\mathcal{H}$
\begin{equation}\label{defpij}
\pi_0:=\mathcal{U}_\Gamma^{-1}\, \hat{\pi}_0\, \mathcal{U}_\Gamma\,,
\end{equation}
which commutes with $H^0$ but is not necessarily a spectral projection of it (unless the range of $\lambda_0$ is isolated from the rest of the spectrum). When a smooth global section $\hat{\phi}_0$ exists, the principal 
{\it Wannier function} $\phi_0$ is defined by:
\beq\label{expl-FBZ}
\phi_0(x) := \big[\mathscr{U}_\Gamma^{-1}\hat{\phi}_0\big](x)\,=
\,|E_*|^{-\frac 12}\int_{\mathbb{T}_*}\hat{\phi}_0(x,\theta)\,d\theta\,
\eeq
and has rapid decay. The family $\phi_\gamma:=\tau_{\gamma}\phi_0$ generated from $\phi_0$ by translations over the lattice $\Gamma$ is an orthonormal basis for the subspace $\pi_0\mathcal{H}$, consisting of (exponentially) localized functions. 

As we have already mentioned, the main spectral result in \cite{CHP} was  obtained under the hypothesis that a global smooth section as in \eqref{defhatpi} exists, or equivalently, that there exists a localized Wannier basis for $\pi_0$. Our main concern in the present paper is to show that roughly the same spectral results can be proved demanding {\bf neither} the global simplicity of $\lambda_0$, {\bf nor} the symmetry with respect to $\xi \mapsto -\xi$ of the symbol of $H^0$.
 In other words, we no longer demand the existence of a localized Wannier basis for the range of $\pi_0\,$.

\subsection{Main statements}
Without any loss of generality we may assume that $\inf(\sigma(H^0))=0\,$. Let us state our only additional hypothesis concerning the bottom of the spectrum of the unperturbed operator $H^0$.

\begin{hypothesis}\label{Hyp-1}
The map $\lambda_{0}:\mathbb{T}_*\rightarrow\mathbb{R}$ has a unique non-degenerate global minimum value realized for $\theta_0\in \mathbb{T}_*$ and $\lambda_0(\theta_0)=0\,$. 
\end{hypothesis}

When $A^\Gamma=0$  the above hypothesis is always satisfied \cite{KS, CHP} with $\theta_0=0\,$ and moreover,  $\lambda_0(\theta)=\lambda_0(-\theta)$. Also, under this hypothesis and if $b$ is small enough, the set 
$$
\Sigma_b:=\lambda_0^{-1}([0,b))\subset \mathbb{T}_*
$$
becomes a simply connected neighbourhood of the minimum.\\

Before we can state the main result of our paper we need some more notation: 
\begin{itemize}
\item Given a vector potential $A$ with components of class
$C^\infty_{\text{\sf pol}}(\X)\,$ and a symbol $F(x,\xi)$ we can define  the following Weyl-magnetic pseudodifferential operator:  
\begin{align}\label{hc111}
&  \big(\mathfrak{Op}^A(F)u\big)(x)\ :=\
(2\pi)^{-2}\int_\X\int_{\X^*}e^{i\langle
\xi,x-y\rangle}e^{-i\int_{[x,y]}A}\, F\left(\frac{x+y}
{2},\xi\right )u(y)\,d\xi\,dy\,.
\end{align}
The magnetic pseudodifferential calculus was developed in \cite{MP1}; for the convenience of the reader we summarized its main features in Section \ref{ss1.3}. 

\item $d_H(M_1,M_2)$ denotes the Hausdorff distance between the subsets $M_1$ and $M_2$ in $\mathbb{R}$. 
\end{itemize}

\begin{theorem}\label{mainTh}  Let $H^{\epsilon,\kappa}$, with $(\epsilon,\kappa)\in[0,1]\times[0,1]$,  be the family  introduced in \eqref{mainH} for a Hamiltonian $H^0$ satisfying Hypothesis \ref{Hyp-1} and a magnetic field of the form \eqref{Bek}.
Then there exists $b>0$ and a smooth real function $\lambdabar^\epsilon:\mathbb{T}_*\rightarrow\mathbb{R}$ 
such that:
\begin{enumerate}
 \item [{\rm(i)}] For any $N\in\mathbb{N}^*$ there exist some constant $C_0>0$ and some
$(\epsilon_0,\kappa_0)\in(0,b/N)\times(0,1)$, such that, for
any $(\epsilon,\kappa)\in(0,\epsilon_0]\times(0,\kappa_0]$,
\beq\label{horiac2}
d_H\left(\sigma\big(H^{\epsilon,\kappa}\big)\bigcap[0,N\epsilon]\,,\,\sigma\big(\mathfrak{Op}^{\epsilon,\kappa}(\widetilde{\lambdabar^\epsilon})\big)\bigcap[0,N\epsilon]\right)\,\leq\,C_0\big(\kappa\epsilon+\epsilon^{2}\big),
\eeq
where $\widetilde{\lambdabar^\epsilon}:\X^*\rightarrow\mathbb{R}$ is the periodic extension of $\lambdabar^\epsilon:\mathbb{T}_*\rightarrow\mathbb{R}$  considered as a function  on $\Xi\,$ constant along the directions in $\X\times\{0\}$, and $\mathfrak{Op}^{\epsilon,\kappa}(\widetilde{\lambdabar^\epsilon})\equiv \mathfrak{Op}^{A^{\epsilon,\kappa}}(\widetilde{\lambdabar^\epsilon})$ is the magnetic quantization of $\widetilde{\lambdabar^{\epsilon}}$ as in \eqref{hc111}.

\item [{\rm(ii)}] For every $m\in \mathbb{N}$ there exists $C_m>0$ such that for all $\theta\in \Sigma_b$ and $|\alpha|\leq m$ we have
$$
\left|\partial^\alpha\lambdabar^\epsilon(\theta)-
\partial^\alpha\lambda_{0}(\theta)\right|\leq C_m\,\epsilon\,.
$$ 
\end{enumerate}
\end{theorem}

\vspace{0.2cm}

Combining the above theorem with the results from \cite{CHP} regarding the existence of spectral gaps for magnetic pseudo-differential operators like  $\mathfrak{Op}^{\epsilon,\kappa}(\widetilde{\lambdabar^\epsilon}) $, we will prove in Subsection \ref{hcfin} the following two corollaries: 

\begin{corollary}\label{C1}
Assume $B^\Gamma=0\,$. Then for any integer $N\geq1$, there exist some constants  $C_0,C_1,C_2 >0\,$, and 
$(\epsilon_0,\kappa_0)\in(0,\widetilde{b}/N)\times(0,1)\,$,  such
that, for
any $(\epsilon, \kappa) \in (0,\epsilon_0]\times (0,\kappa_0]$, there exist $a_0<b_0<a_1<\cdots <a_N<b_N$,
with $a_0= \inf\{\sigma(H^{\epsilon,\kappa})\}$ so that:
\begin{align}
&\sigma(H^{\epsilon,\kappa})\cap \left [a_0,b_N\right ]\subset \bigcup_{k=0}^N [a_k,b_k]\,,\quad{\rm dim}
\big({\rm Ran} E_{[a_k,b_k]}(H^{\epsilon,\kappa}) \big )= + \infty\,,\nonumber \\
& b_k-a_k\leq C_0 \,\epsilon\big(\kappa + C_1\, \epsilon^{1/3}\big)\, , \; 0\leq k\leq N\,,\quad {\rm and}\quad  a_{k+1}-b_k\geq\frac{1}{C_2}  \epsilon \, ,\; 0\leq k\leq N-1\,.\label{feb-19}
\end{align}
\end{corollary}

\begin{corollary}\label{C2}
 Assume $B^\Gamma\neq 0\,$. Under  Hypothesis \ref{Hyp-1}, for any integer $N\geq1$, there exist positive constants  $C_0,C_1,C_2 \,$, and 
$(\epsilon_0,\kappa_0)\in(0,\widetilde{b}/N)\times(0,1)$,  such that for
any $(\epsilon,\kappa)\in (0,\epsilon_0]\times (0,\kappa_0]$ there exist $a_0<b_0<a_1<\cdots <a_N<b_N$,
with $a_0= \inf\{\sigma(H^{\epsilon,\kappa})\}$ so that:
\begin{align}\label{feb-19-bis}
&\sigma(H^{\epsilon,\kappa})\cap \left [a_0,b_N\right ]\subset \bigcup_{k=0}^N [a_k,b_k]\,,\quad{\rm dim}
\big({\rm Ran} E_{[a_k,b_k]}(H^{\epsilon,\kappa}) \big )= + \infty\,, \\
& b_k-a_k\leq C_0\epsilon \big(\kappa + C_1\, \epsilon^{1/5}\big) \mbox{ for }0\leq k\leq N\,, \mbox{ and }  a_{k+1}-b_k\geq\frac{1}{C_2}  \epsilon \,, \mbox{ for }  0\leq k\leq N-1\,.\nonumber 
\end{align}
\end{corollary}

As a final remark, let us notice that in the case when the global minimum of the band $\lambda_{0}$ is attained at several non degenerate   points, one can repeat the argument given  in Section 4 of \cite{CHP} making cut-offs around each minimum and the results remain true.

\subsection{A roadmap of the proof} \label{ss2.2}
 
\begin{enumerate}
 \item We start by constructing a global smooth section $\hat{\psi}_0:\mathbb{T}_*\rightarrow\mathscr{F}$ of norm 1 vectors in $\mathscr{F}_\theta$ that coincides with the Bloch eigenvector $\hat{\phi}_0(\theta)$ on the neighborhood $\Sigma_b$ of the global minimum $\theta_0$ of $\lambda_0$. Then the range of its associated orthogonal projection $\pi:=\mathcal{U}_\Gamma^{-1} \, \hat{\pi}\,\mathcal{U}_\Gamma$ in $\mathcal{H}$ will have an orthonormal basis of localized functions given by all the  $\Gamma$-translations of $\psi_0:=\mathscr{U}_\Gamma^{-1}\hat{\psi}_0$.  This will be presented in Section \ref{S_q-W}. An important observation is that although the projection $\pi$ does not commute with $H^0$, it still contains information about the lowest part of the spectrum of $H^0$. This is because the bounded operator $\pi H^0\pi$ has exactly one Bloch band $\widetilde{\lambda}_0(\theta)$ which coincides with $\lambda_0(\theta)$ on $\Sigma_b$ (a fact proved in Proposition \ref{mod-Bloch-func}).

\item  We notice that the operator $\pi^\perp H^0\pi^\perp$ seen as acting in  $\pi^\perp\mathcal{H}$ is bounded from below by $b$. Up to a use of the Feshbach-Schur argument (see for example \cite{GS}) and keeping the spectral parameter $E$ in the interval $[0,E_0]$ with $E_0<b/2\,$, the operator $H^0-E$ is invertible if and only if the reduced operator 
\begin{equation}\label{defS}
S(E):=\pi (H^0-E)\pi -\pi H^0 \pi^\perp [\pi^\perp (H^0-E)\pi^\perp]^{-1}\pi^\perp H^0 \pi
\end{equation}
is invertible in $\pi \mathcal{H}\,$.\\
 Since we are only interested in values of $E$ which are close to zero, we may expand the reduced resolvent around $E=0$ and obtain:
\begin{align*}
S(E)&=\pi H^0\pi -\pi H^0 \pi^\perp [\pi^\perp H^0\pi^\perp]^{-1}\pi^\perp H^0 \pi -E (\pi+\pi H^0 \pi^\perp [\pi^\perp H^0\pi^\perp]^{-2}\pi^\perp H^0 \pi)\\
&+\mathcal{O}(E_0^2)\,.
\end{align*}
We observe that the operator coefficient of  $E$, i.e. 
\begin{equation}\label{defY}
Y:=\pi+\pi H^0 \pi^\perp [\pi^\perp H^0\pi^\perp]^{-2}\pi^\perp H^0 \pi
\end{equation}
is bounded and satisfies  $Y\geq \pi$. Thus $S(E)$ is invertible in $\pi{H}$ iff $Y^{-1/2}S(E) Y^{-1/2}$ is invertible in $\pi\mathcal{H}$, and some elementary arguments (see the proof of Proposition~\ref{bernard7'}) allow to obtain an estimation of the Hausdorff distance between  the spectrum of the 'dressed' operator 
\begin{equation}\label{deftildeH}
\widetilde{H}:=Y^{-1/2}\{\pi H^0\pi -\pi H^0 \pi^\perp [\pi^\perp H^0\pi^\perp]^{-1}\pi^\perp H^0 \pi\}Y^{-1/2}
\end{equation}
restricted to the interval $[0,E_0]$ and the spectrum of $H^0$ in the same interval.  

We notice that due to our construction of  a special Wannier-like basis  for the range of $\pi$, the 'dressed' operator $\widetilde{H}$ can be identified with a matrix acting in $\ell^2(\Gamma)$. The details will be presented in Section \ref{S.4} as a consequence of a more abstract theorem.

 \item The crucial step (explained below) is to extend these objects to the magnetic case and construct a magnetic matrix acting in  $\ell^2(\Gamma)$, whose spectrum has not only gaps, but it also lies close enough to the spectrum of the full operator ${H}^{\epsilon,\kappa}$ near zero such that some gaps must also appear in the spectrum of ${H}^{\epsilon,\kappa}$. In Section \ref{S.5} we use the procedure developed in \cite{Ne-RMP,HS,CHN,CIP} in
order to define a \textit{magnetic quasi-band projection}
$\pi^{\epsilon,\kappa}$ which is the magnetic version of $\pi$ and its range is again spanned by some localized functions which are indexed by $\Gamma$. An important observation is that $(\pi^{\epsilon,\kappa})^\perp{H}^{\epsilon,\kappa}(\pi^{\epsilon,\kappa})^\perp$ is still bounded from below by $b/2$ if $\epsilon$ is small enough.

\item We can repeat the above Feshbach-Schur argument for the pair $\big(H^{\epsilon,\kappa},\pi^{\epsilon,\kappa}\big)$ with $E_0=N\epsilon$ where $N$ is a fixed large number and construct a dressed 
\textit{magnetic matrix} $\widetilde{H}^{\epsilon,\kappa}$ acting in $\ell^2(\Gamma)$, 
such that its spectrum in the interval $[0,N\epsilon]$ is at a Hausdorff distance of order $\epsilon^2$ from
the spectrum of the full operator ${H}^{\epsilon,\kappa}$. Hence if we can prove that the matrix $\widetilde{H}^{\epsilon,\kappa}$ has gaps or order $\epsilon$ in the interval $[0,N\epsilon]$, the same must be true for $H^{\epsilon,\kappa}$.

\item In Section \ref{modif-B-function} we  construct a 'quasi-band' periodic and smooth function $\widetilde{\lambdabar^\epsilon}(\theta)$ which is close to $\lambda_0$ on $\Sigma_b$ such that the Hausdorf distance between the spectrum of $\widetilde{H}^{\epsilon,\kappa}$ (acting on $\ell^2(\Gamma)$) and the spectrum of the magnetic quantization $\mathfrak{Op}^{\epsilon,\kappa}(\widetilde{\lambdabar^\epsilon})$ (acting on $L^2(\mathcal{X})$) is of order $\epsilon \kappa$. This leads to \eqref{horiac2}. 
\end{enumerate}

\section{The quasi Wannier system}
\label{S_q-W}

Let us spell out  a straightforward but  important consequence of  Hypothesis \ref{Hyp-1}:\\
\begin{lemma}\label{newlemma}
There exists $\widetilde{b}>0$  such that, for every $0<b\leq\widetilde{b}$,  we have:
\begin{itemize}
\item The set $
\Sigma_b$
is diffeomorphic to the open unit disc in $\mathbb{R}^2$, has a smooth boundary and contains $\theta_0\,$. 
\item The function $\lambda_0$ is smooth on $\Sigma_b$ and its Hessian matrix is positive. 
\item For $\theta$ outside of $\Sigma_b$ we have  $\hat{H}^0(\theta)\geq b\,$.  
\end{itemize}
\end{lemma}

Let us fix some $b\in(0,\widetilde{b})$ and consider the local smooth section $\left.\hat{\phi}_0\right|_{\overline{\Sigma}_b}:\overline{\Sigma}_b\rightarrow\left.\mathscr{F}\right|_{\overline{\Sigma}_b}$.  

\begin{proposition}
There exists a global smooth section of norm 1 vectors
\beq
\hat{\psi}_0:\mathbb{T}^*\rightarrow\mathscr{F}\,,
\eeq
such that $\hat{\psi}_0(\theta)=\hat{\phi}_0(\theta)$ for any $\theta\in\Sigma_b\,$.
\end{proposition}
\begin{proof}
Let $f_j\in C_0^\infty(\mathcal{X})$, $j\in\{1,2\}$, with unit norms in $L^2(\mathcal{X})$ such that both of them have support in $|x|\leq 1/10$ and moreover, $f_1(x)f_2(x)=0$ for all $x$. Let us define
$$\hat{f_j}(x,\theta)=\sum_{\gamma\in \Gamma} e^{i<\theta, \gamma>} f_j(x-\gamma).$$
These functions are smooth in $x$ and belong to $\mathcal{F}_\theta$. We can check that $\hat{f_1}(x,\theta)\hat{f_2}(x,\theta)=0$ for all $x\in E$, thus their scalar product in $\mathcal{F}_\theta$ equals zero. Moreover, both have norm one at fixed $\theta$.

When $\theta$ is close to $\theta_0$, the eigenvector $\hat{\phi}(x,\theta)$ can be chosen to be smooth as a function of $x$ due to elliptic regularity. It is also smooth as a function of $\theta$ on a small neighborhood of $\theta_0$. Now the vector $\hat{\phi}(\cdot,\theta_0)$ is certainly not parallel with both $\hat{f_j}(\cdot ,\theta_0)$ and there must exist a $j$ (assume without loss of generality that $j=1$) such that 
$$\vert \langle \hat{\phi}(\theta_0),\hat{f_1}(\cdot ,\theta_0)\rangle\vert \leq 1/\sqrt{2}\,.$$
 Due to continuity in $\theta$, there exists a small ball $B_r(\theta_0)\subset \overline{\Sigma}_b$ where we have 
 
 $$\vert \langle \hat{\phi}(\theta),\hat{f_1}(\cdot ,\theta)\rangle\vert \leq 3/4\,,\quad \forall \theta\in B_r(\theta_0)\,.$$

Let $0\leq g(\theta)\leq 1$ with support in $B_r(\theta_0)$, equal to one on $B_{r/2}(\theta_0)$. Define $$\hat{h}(x,\theta):= g(\theta) \hat{\phi}(x,\theta) +(1-g(\theta))\hat{f}(x,\theta)\,.$$ Then $\Vert \hat{h}(\cdot,\theta)\Vert^2\geq 1/8$ hence $\hat{\psi}_0(x,\theta):=\hat{h}(x,\theta)/\Vert \hat{h}(\cdot,\theta)\Vert\in \mathcal{F}_\theta$ is a smooth extension of $\hat{\phi}$ in both $x$ and $\theta$. 
\end{proof}

We are now ready to define the principal {\it quasi Wannier function} $\psi_0$ by the formula
\beq \label{pseudo-W-f} 
\psi_{0}:=\mathcal{U}_\Gamma^{-1}\hat{\psi}_{0}\in L^2(\X)\,,
\eeq
where $\mathcal U_\Gamma$ has been introduced in \eqref{UGamma}.
The above function has norm $1$ in $\mathcal{H}$, it is smooth on $\X$ and has fast decay together with all its derivatives. Also, the corresponding family
\beq\label{R-q-W-basis}
\psi_\gamma:=\tau_{\gamma}\psi_0\in
L^2(\mathcal{X})\,,\qquad \gamma\in\Gamma\,,
\eeq
forms an orthonormal system in $L^2(\mathcal{X})$ that we shall call the \textit{quasi Wannier system associated with  the energy window $[0,b]$}. We shall also  denote by
$\pi\in\mathcal L(\mathcal{H})$ the orthogonal projection 
on the closed linear subspace generated by the quasi Wannier system: 
\beq\label{horiac5}
\mathcal{H}_0:=\overline{{\rm Span}\{\psi_\gamma:\; \gamma\in \Gamma\}}\subset L^2(\mathcal{X}).
\eeq
We shall use the notation:
\beq\label{hatpi}
\hat{\pi}(\theta):=|\hat{\psi}_0(\theta)\rangle\langle\hat{\psi}_0(\theta)|.
\eeq
The following statement can be proven by the same arguments as Proposition 3.12 of \cite{CHP}. 

\begin{proposition}\label{free-ps-W}
The vector $\psi_0$ defined in \eqref{pseudo-W-f} belongs to the Schwartz space $\mathcal{S}(\mathcal{X})$. Moreover, if we consider the Weyl symbol  $p_0$ in $\mathscr{S}(\Xi)$  of the orthogonal
projection $|\psi_0\rangle\langle\psi_0|$
$$
p_0(x,\xi)=(2\pi)^{-1}\int_\mathcal{X}e^{i<\xi,v>}\psi_0(x+(v/2))
\overline{\psi_0(x-(v/2))}\,dv\,,
$$
then the series $\underset{\gamma\in\Gamma}{\sum}p_0(x-\gamma,\xi)$
converges pointwise with all its
derivatives to a $\Gamma$-periodic symbol $p\in S^{-\infty}(\Xi)$  and we have 
\beq \label{defpi}
 \pi =\mathfrak{Op}(p)\,,\,  K_p(x,y):=\sum_{\gamma\in \Gamma}\psi_0(x-\gamma)\overline{\psi_0(y-\gamma)}\,, \quad 
\mathscr{U}_\Gamma\pi\mathscr{U}_\Gamma^{-1}=\int^\oplus_
{\mathbb{T}_*} |\hat{\psi}_0(\theta)\rangle \langle \hat{\psi}_0(\theta)|d\theta\,,
\eeq
where $K_p$ denotes the distribution kernel of $\pi$.
\end{proposition}

\begin{proposition}\label{bd-Ham0-band}
With $H^0$ introduced in \eqref{H0} and $\pi$ in \eqref{defpi}, 
 we have  $\pi\mathcal{H}\subset\mathcal{D}(H^0)$, 
$H^0\pi\in\mathcal L(\mathcal{H})$ and $\pi H^0$ has a bounded closure.
\end{proposition}
\begin{proof}
Due to the properties of $\psi_0$ listed after its definition in \eqref{pseudo-W-f} and due to the fact that $H^0$ is a differential operator with polynomially bounded coefficients, the quasi Wannier function $\psi_0$ belongs to $\mathcal{D}(H^0)$. Moreover, due to the rapid decay of the quasi Wannier function and its derivatives, it follows that $\mathcal{H}_0$ (defined in \eqref{horiac5})
belongs in fact to $\mathcal{D}(H^0)$. We conclude that $\mathcal{H}_0=\pi\mathcal{H}\subset\mathcal{D}(H^0)$ and thus $H^0\pi$ is a well defined bounded linear operator on $\mathcal{H}$. Moreover, the operator $\pi H^0:\mathcal{D}(H^0)\rightarrow\mathcal{H}$ is a restriction of the adjoint $\big(H^0\pi\big)^*$ and thus has a bounded closure in $\mathcal{H}\,$.
\end{proof}

\section{An abstract reduction argument}\label{S.4}
We consider a more abstract setting and give the details of 
the arguments sketched in the third item of Subsection \ref{ss2.2}  that gives a procedure to apply the Feshbah-Schur method in a more complicate situation when the commutator $\big[H,\pi\big]$ is no longer controlled by a small parameter (as was the situation in our previous paper \cite{CHP}).
\begin{hypothesis}\label{Prop-FS}
We consider a triple $(H,\Pi,\beta)$ where $H$ is  a positive
self-adjoint operator,   $\Pi$ is an orthogonal projection such that 
$H\Pi$ is bounded (as well as $\Pi H$), $\beta >0$ and
\begin{equation}\label{eq:lb}
\Pi^\bot H\Pi^\bot\geq 2\beta \Pi^\bot\,.
\end{equation}
\end{hypothesis}
This implies that $\Pi^\bot (H-E)\Pi^\bot$ is invertible  in 
$\Pi^\bot\mathcal{H}$  for 
$E\in[0,2\beta)$ and we denote by $R_\bot(E)\in\mathcal L(\Pi^\bot\mathcal{H})$ its inverse. The spectral theorem gives:
\beq\label{E0}
  \sup_{E\in [0,\beta ]}\Vert R_\bot(E)\Vert \leq \beta^{-1}.
\eeq
Starting from
\beq\label{dress}
Y:=\Pi+\Pi H\Pi^\bot R_\bot(0)^2\Pi^\bot
H\Pi\,,
\eeq
and noting  that
$$
  Y\geq \Pi\,,
$$
we introduce
\beq\label{horiac6}
\widetilde{H}:= Y^{-1/2}\left[\Pi H\Pi-\Pi H\Pi^\bot R_\bot(0)\Pi^\bot
H\Pi\right]Y^{-1/2}\in\mathcal L(\Pi\mathcal{H}).
\eeq
\begin{proposition}\label{bernard7'}
For any $\beta' < \beta$ we have:
\beq
d_H\{\sigma(H)\cap [0,\beta'],\sigma(\widetilde{H})\cap [0,\beta']\}\leq 
\, \| H\Pi \|^2\ (\beta')^2 \beta^{-3}\,,
\eeq
where $\widetilde H$ is defined in \eqref{horiac6}.
\end{proposition}

\begin{proof}[Proof]
Using the Feshbach-Schur reduction
 we get that if $E\in[0,\beta]$ then 
the
operator $H-E$ is invertible
\textit{if and only if} the operator
\beq\label{def-S}
S(E):=\Pi \big(H -E\big)\Pi \,-\,
\Pi  H  R_\bot(E)H \Pi
\eeq
is invertible in $\Pi \mathcal{H}$. In this case we have the identity:
\beq\label{horiac7}
\Pi(H-E)^{-1}\Pi\ =S(E)^{-1}.
\eeq
Using the resolvent
equation, we have, if $E\in [0,\beta')$,
\begin{align*}
&R_\bot(E)=R_\bot(0)+R_\bot(0)^2+
X(E)\,,\quad X(E):=E^2R_\bot(0)^2R_\bot(E)\,,\\
  &\|X(E)\|\leq E^2\beta^{-3}\leq \beta'^2\beta^{-3}\,.
\end{align*}
Thus we can write $S(E)$ as:
\begin{align*}
S(E)&=\Pi H\Pi-\Pi  H  R_\bot(0)H \Pi -E(Y-\Pi)-E\Pi+\Pi HX(E)H\Pi\\
&=\Pi\left[H-H  R_\bot(0)H\right]\Pi-EY+\Pi HX(E)H\Pi.
\end{align*}
Since $Y$ in \eqref{dress} is bounded from below by $1$, we can
define the positive operator $Y^{-1/2}$ in $\mathcal L(\Pi\mathcal{H})$. 
Then (use the notation \eqref{horiac6} in the above identity) $S(E)$ is 
invertible \textit{if and only if}
\begin{align}\label{horiac8}
Y^{-1/2}S(E)Y^{-1/2}=
\widetilde{H}-E+Y^{-1/2}\Pi HX(E)H\Pi Y^{-1/2}
\end{align}
is invertible in $\Pi \mathcal{H}$ with bounded inverse.
We are now prepared to prove Proposition  \ref{bernard7'} and we  do it in two 
steps.

{\bf Step 1.} Assume that $E\in   [0,\beta']$ is in the resolvent set of 
$\widetilde{H}$.
Then \eqref{horiac8} implies:
$$S(E)=Y^{1/2}\{\bb1 +Y^{-1/2}\Pi HX(E)H\Pi 
Y^{-1/2}(\widetilde{H}-E)^{-1})\}(\widetilde{H}-E)Y^{1/2}.$$
We have:
$$\|Y^{-1/2}\Pi HX(E)H\Pi Y^{-1/2}(\widetilde{H}-E)^{-1}\|\leq 
\frac{\beta'^2\|H\Pi\|^2\beta^{-3}}{{\rm 
dist}(E,\sigma(\widetilde{H}))},\mbox{ if }  E\in [0,\beta'].$$
Thus, if ${\rm 
dist}(E,\sigma(\widetilde{H}))>\beta'^2\|H\Pi\|^2\beta^{-3}\,$,  $S(E)$ 
is invertible, hence $E$ is in the resolvent set of $H$. In other words, 
no element of $\sigma(H)\cap [0,\beta']$ can be situated at a distance 
which is larger than $\beta'^2\|H\Pi\|^2\beta^{-3}$ from  
$\sigma(\widetilde{H})\cap [0,\beta')$.

{\bf Step 2.}  Assume that $E\in [0,\beta')$ is in the resolvent set of $H$ hence 
\eqref{horiac7} holds. From \eqref{horiac8} we see that 
$\widetilde{H}-E$ is invertible when the operator:
$$S(E)-\Pi HX(E)H\Pi=\{\bb1-\Pi HX(E)H\Pi (H-E)^{-1}\Pi\}S(E)$$
is invertible. \\
 As before we have:
$$\|\,\Pi HX(E)H\Pi (H-E)^{-1}\Pi\,\|\leq 
\frac{\beta'^2\|H\Pi\|^2\beta^{-3}}{{\rm dist}(E,\sigma(H))},\quad E\in 
[0,\beta').$$
Hence $\widetilde{H}-E$ is invertible when ${\rm 
dist}(E,\sigma(H))>\beta'^2\|H\Pi\|^2\beta^{-3}$. Equivalently, no 
element of $\sigma(\widetilde{H})\cap [0,\beta')$ can be situated at a 
distance larger than $\beta'^2\|H\Pi\|^2\beta^{-3}$ from  $\sigma(H)\cap 
[0,\beta']$.

\end{proof}

\begin{corollary}\label{bernard1}
Let $0<D_1<D_2<\beta' <\beta$  and assume $(D_1,D_2) \cap 
\sigma(\widetilde{H})=\emptyset$.
Then, if $$\| H\Pi \|^2\ (\beta')^2 \beta^{-3} < \frac 12 (D_2-D_1)
$$
we have
$$
(D_1 + \| H\Pi \|^2\ (\beta')^2 \beta^{-3} \,, D_2- \| H\Pi \|^2\ 
(\beta')^2 \beta^{-3}) \cap \sigma (H) = \emptyset \,.
$$
\end{corollary}
It will be interesting to apply this corollary for a family indexed by 
$\eta \in [0,\epsilon_1]$ of triples $(H(\eta), \Pi (\eta), \beta)$ satisfying 
\eqref{eq:lb}. This leads to:
\begin{corollary}\label{bernard2}
Let $0<D_1<D_2<\beta' <\beta$  and assume $(D_1,D_2) \cap 
\sigma(\widetilde{H} (\eta))=\emptyset\,$, for all $\eta \in [0,\epsilon_1]\,$.
Then, if $$D:=\sup_{\eta \in [0,\epsilon_1]} \| H (\eta) \Pi (\eta)  \|^2\ (\beta')^2 
\beta^{-3} < \frac 12 (D_2-D_1)\,,
$$
we have
$$
(D_1 +D\,, D_2-D ) \cap \sigma (H(\eta)) = 
\emptyset \,.
$$
\end{corollary}
\begin{proof}
It is a direct consequence of Corollary \ref{bernard1}\,. 
\end{proof}

Finally, the next proposition gives sufficient conditions for the appearance of gaps of order $\eta$ in the spectrum of $H(\eta)$. 

\begin{proposition}\label{C-FS}
Suppose that that there exist two numerical constants $0<C_1<C_2$ such that the interval $(C_1,C_2)$ belongs to the resolvent set of $\eta^{-1}\widetilde{H}(\eta)$ for all $0< \eta \leq \min(\epsilon_1,\frac{\beta}{C_2+1})$. 

Let $C_0:=(C_2+1)^2 (\sup_{\eta\in 
[0,\epsilon_1]}\|H (\eta)\Pi (\eta)\|)^2 
\beta^{-3}$ and $\epsilon_0 := \min 
(\epsilon_1,\frac{\beta}{C_2+1}, \frac 12 (C_2-C_1)/C_0)\,.$ 
Then, 
$$(C_1\eta 
+C_0\eta^2, C_2\eta-C_0\eta^2)\cap \sigma(H(\eta))=\emptyset\,, \quad \forall\eta\in (0,\epsilon_0)\,.$$
\end{proposition}
\begin{proof}
Let $D_1 = C_1 \eta\,$, $D_2=C_2 \eta\,$, $\beta' = (C_2+1)\eta\,$. We then apply Corollary \ref{bernard2}.
\end{proof}

\section{The magnetic pseudodifferential calculus}\label{ss1.3}
Let us recall from \cite{CHP} the following two notations:
\begin{equation}\label{defLambdaOmega}
{\Lambda}^A(x,y):=e^{-i\int_{[x,y]}A},\qquad\Omega^B(x,y,z):=e^{-i\int_{
<x,y,z>} \hspace*{-3pt}B}\,,
\end{equation}
where $[x,y]$ denotes the oriented interval from $x$ to $y$ and $<x,y,z>$ denotes the oriented triangle with vertices $x,y,z$, with the integrals being the usual invariant integrals of $k$-forms for $k=1,2$. \\
We shall also use the following shorthand notation:
\beq\label{sh-not-magn-2}
\Lambda^\epsilon\equiv\Lambda^{\epsilon A^0},\
\Lambda^{\epsilon,\kappa}\equiv\Lambda^{A^{\epsilon,\kappa}},\
\Omega^\epsilon\equiv\Omega^{\epsilon B_0},\
\Omega^{\epsilon,\kappa}\equiv\Omega^{B^{\epsilon,\kappa}},\ 
\widetilde{\Lambda}^{\epsilon,\kappa}\equiv\Lambda^{\kappa A^\epsilon},\ 
\widetilde{\Omega}^{\epsilon,\kappa}\equiv\Omega^{\kappa\epsilon B_\epsilon}\,.
\eeq

We notice for future use  that 
\beq\label{D-Omega}
{\Lambda}^A(x,z){\Lambda}^A(z,y){\Lambda}
^A(y,x)\,=\,\Omega^B(x,z,y)\,,
\eeq
and that there exists $C >0$ such that,  for any magnetic field $B$ of class
$BC^\infty(\X)\,$, 
\beq\label{Est-Omega}
\left|\Omega^B(x,y,z)-1\right|\leq C\, \|B\|_\infty
|(y-x)\wedge(z-x)|\,.
\eeq
We notice also that for our choice of gauge for $A_0$ we have the relations:
\beq
\Lambda^\epsilon(x,y)=\exp\left\{i(B_0/2)\,x\wedge y\right\}
\eeq
\beq\label{defOmega}
\Omega^\epsilon(x,y,z)=\exp\left\{iB_0\, (x-y)\wedge (z-y)\right\}.
\eeq

Let us recall the gauge covariance property which states that whatever magnetic potential $A$ we choose for a given magnetic field $B(x):=B_{12}(x)=-B_{21}(x)$ we have:
\begin{align}\label{26-02-2}
\{-i\nabla-A(x)\}^2\Lambda^A(x,z)=\Lambda^A(x,z)\left\{-i\nabla\,-\,a(x,z)
\right\}^2,
\end{align}
where
\beq\label{axz}
a_j(x,z)\,=\,\underset{k}{\sum}(x-z)_k\int_0^1
B_{jk}\big(z+s(x-z)\big)s\,ds\, \mbox{ for } j=1,2\,.
\eeq

Let $\X^*\cong\mathbb{R}^2$ be the dual of our configuration space $\X\equiv\mathbb{R}^2$ and denote by $<\cdot\,,\,\cdot>$ the duality map $\X^*\times\X\rightarrow\mathbb{R}$ between these two spaces. Let $\Xi:=\X\times\X^*$ be the associated phase space with the canonical symplectic structure 
\beq
\sigma\big((x,\xi),(y,\eta)\big):=<\xi,y>-<\eta,x>,\quad\forall\big((x,\xi),(y,\eta)\big)\in\Xi\times\Xi\,.
\eeq

We will use the following class of H\"{o}rmander type symbols. 
\begin{definition}\label{Def-symb}
For any
$s\in\mathbb{R}$ we denote by
$$
S^s(\Xi):=\{F\in
C^\infty(\Xi)\mid\nu^{s}_{n,m}(F)<+\infty\,,\forall(n,m)\in\mathbb{N}
\times\mathbb{N}\}\,,
$$
where
$$
\nu^{s}_{n,m}(f):=\underset{(x,\xi)\in\Xi}{\sup}\ \underset{|\alpha|\leq n}{\sum}\ 
\underset{|\beta|\leq m}{\sum}\left|\langle \xi\rangle
^{-s+m}
\big(\partial^\alpha_x\partial^\beta_\xi f\big)(x,\xi)\right|\,,
$$ 
and
\beq\label{M-inv}
S^-(\Xi):=\underset{s<0}{\bigcup}S^s(\Xi)\quad \mbox{
 and }\quad  S^{-\infty}(\Xi):=\underset{s\in\mathbb{R}}{\bigcap}S^s(\Xi)\,.
 \eeq
\end{definition}
\begin{definition}
A symbol $F$ in $S^s(\Xi)$ is called {\it
elliptic}  
if there exist two positive constants $R$ and $C$ such that $$|F(x,\xi)|\geq
C \, \langle \xi\rangle ^s\,,$$ 
for any
$(x,\xi)\in\Xi$ with $|\xi|\geq R\,$.
\end{definition}
\begin{definition}
We consider a vector potential $A$ with components of class
$C^\infty_{\text{\sf pol}}(\X)\,$. For any $F\in\mathscr{S}(\Xi)$ we can define (\cite{MP1})  the following linear operator:  
\begin{align}\label{octhc1}
& u\mapsto \big(\mathfrak{Op}^A(F)u\big)(x)\ :=\
(2\pi)^{-2}\int_\X\int_{\X^*}e^{i\langle
\xi,x-y\rangle}e^{-i\int_{[x,y]}A}\, F\left(\frac{x+y}
{2},\xi\right )u(y)\,d\xi\,dy\,,\\ 
& \forall u\in \mathscr{S}(\X).\nonumber
\end{align}
\end{definition}

This operator is continuous on $ \mathscr{S}(\X)$
 and has a natural extension by duality to
$\mathscr{S}^\prime(\X)$. For $A=0$ we obtain the usual Weyl calculus.  In \cite{MP1} it is proven that
\begin{proposition}\label{MP} 
 For any vector potential $A$ with components of class $C^\infty_{\text{\tt pol}}(\X)$ the quantization $\mathfrak{Op}^A$ defines a linear and topological isomorphism between $\mathscr{S}^\prime(\Xi)$ and the linear topological space of continuous operators from $\mathscr{S}(\X)$ to $\mathscr{S}^\prime(\X)$.
\end{proposition} 

We shall denote by $\mathfrak{S}^A$ its inverse, i.e. the application that associates to a given operator the tempered distribution which is its symbol for the magnetic pseudodifferential calculus defined by the vector potential $A$:
\begin{equation}\label{defsymb}
\mathfrak{S}^A \circ \mathfrak{Op}^A = I\,.
\end{equation}

In the same spirit as the Calderon-Vaillancourt theorem for classical pseudo-differential operators, Theorem 3.1 in \cite{IMP1} states that any symbol $F\in S^0_0(\Xi)$ defines a
bounded operator $\mathfrak{Op}^A(F)$ in $L^2(\X)$ with an upper bound of  the operator norm
given by some  seminorm of $F$ as H\"{o}rmander type symbol.  We
denote by $\|F\|_B$ the operator norm of $\mathfrak{Op}^A(F)$ in $\mathcal L
(L^2(\X))$ :
\begin{equation}\label{defFb}
\|F\|_B:= ||\mathfrak{Op}^A(F)||_{\mathcal L (L^2(\X))}\,.
\end{equation}
This norm  only depends on the
magnetic field $B$ and not on the choice of the vector potential (different
choices being unitary equivalent).

We notice that
\beq\label{symbol-h}
H^0\,=\,\mathfrak{Op}^{A^\Gamma}(h),\qquad
 h(x,\xi):=\xi^2+V_\Gamma(x)\,. 
 \eeq
 and for the perturbed Hamiltonians we will use the following shorthand notations
 \beq
H^{\epsilon,\kappa}=\mathfrak{Op}^{A^{\epsilon,\kappa}}(h)=:\mathfrak{Op}^{\epsilon,\kappa}(h),\qquad
H^{\epsilon}=\mathfrak{Op}^{A^{\epsilon}}(h)=:\mathfrak{Op}^{\epsilon}(h).
 \eeq

We also recall from \cite{MP1} that for any two test functions $f$ and $g$ in $\mathscr{S}(\Xi)$ the product of the linear
operators $\mathfrak{Op}^A(f)$ and $\mathfrak{Op}^A(g)$ induces a {\it twisted
Moyal product}, also called magnetic Moyal product, such that
$$
\mathfrak{Op}^A(f)\, \mathfrak{Op}^A(g)=\mathfrak{Op}^A(f \, \,\sharp^B\,
\,g)\,.
$$
This product depends only on the magnetic field $B$ and is  given by an explicit oscillating
integral. The following relation proven in \cite{MP1} (Lemme 4.14):
\beq
\int_{\Xi}dX\big(f\sharp^Bg\big)(X)\,=\,\int_{\Xi}dXf(X)g(X)
\eeq
allows to extend the magnetic Moyal product to a class of tempered distributions that leave  $\mathscr{S}(\Xi)$ invariant by magnetic Moyal composition. In \cite{MP1} it is proven that this class forms a *-algebra that we call the magnetic Moyal algebra and denote by $\mathfrak M^B$. Moreover this class  contains all the H\"{o}rmander type symbols (see \cite{MP1,IMP1}) and the usual {\it composition of symbols theorem} is still valid
(Theorem 2.2 in \cite{IMP1}).

For any invertible symbol $F$ in $\mathfrak M^B$, we denote by $F^-_B$ its inverse. It is shown in
Subsection~2.1 of \cite{MPR1} that, for  any $m>0$ and  for $a>0$ large enough (depending on $m$)
the symbol
$\mathfrak{s}_m(x,\xi):=<\xi>^m+a$, has an inverse in $\mathfrak M^B$. We shall use the shorthand notation 
$\mathfrak{s}^B_{-m}$ instead of $\big(\mathfrak{s}_m\big)^-_B$ and extend it
to any $m\in\mathbb{R}$ (thus for $m>0$ we have
$\mathfrak{s}^B_{m}\equiv\mathfrak{s}_{m}$).
\\
The following results have been established in \cite{IMP2} (Propositions 6.2 and 6.3):
\begin{proposition}\label{IMP}~
\begin{enumerate}
\item 
If $F\in S^0(\Xi)$ is invertible in $\mathfrak M^B$, then the inverse
$F^-_B$ also belongs to $S^0(\Xi)\,$.
\item  For $m<0$, if $f\in S^{m}(\Xi)$ is
such that $1+f$ is invertible in $\mathfrak M^B$, then $(1+f)^-_B-1\in
S^m(\Xi)\,$.
\item 
 For $m>0$, if $G\in
S^m(\Xi)$ is invertible in $\mathfrak M^B$, with $\mathfrak{Op}^A\big(\mathfrak
s_{m}\, \sharp^BG^{-}_B\big)\in\mathcal{L}\big(L^2(\X)\big)\,$, then $G^-_B\in
S^{-m}(\Xi)\,$.
\end{enumerate}
\end{proposition}
From the explicit form of the magnetic Moyal product (see in \cite{MP1}) we easily notice the following fact.
\begin{proposition}\label{Gamma*per}
Let us consider the lattice $\Gamma_*\subset\X^*$. If $f$ and $g$ are $\Gamma_*$-periodic symbols, then their magnetic Moyal product is also $\Gamma_*$-periodic.
\end{proposition}
\paragraph{About the resolvent}~\\
By Theorem 4.1 in \cite{IMP1},  for any real elliptic symbol
$h\in S^m(\Xi)$ (with $m>0$) and   for any $A$ in $C^\infty_{\text{\sf
pol}}(\X,\mathbb R^2)$, the operator
$\mathfrak{Op}^A(h)$ has a closure $H^A$ in $L^2(\X)$ that is self-adjoint on a domain
$\mathcal{H}^m_A$ (a {\it magnetic Sobolev space}) and lower semibounded. Thus
we can define its resolvent $(H^A-\z)^{-1}$ for any $\z\notin\sigma(H^A)$ and
Theorem 6.5 in \cite{IMP2} states that it exists a symbol
$r^B_\z(h)\in S^{-m}(\Xi)$ such that 
$$
(H^A-\z)^{-1}\ =\ \mathfrak{Op}^A(r^B_\z(h))\,.
$$
\paragraph{Integral kernels and symbols}~\\
For symbols of class $S^0(\Xi)$, we have seen that the
associated magnetic pseudodifferential operator is bounded in
$\mathcal{H}$ and is self-adjoint if and only if its symbol is real. In that
case we can also define its resolvent and the results in \cite{IMP2}, cited
above, show that it is also defined by a symbol of class $S^0(\Xi)\,$.\\

Given any tempered distribution $T\in\mathscr{S}^\prime(\X\times\X)$ we shall denote by $\Int(T):\mathscr{S}(\X)\rightarrow\mathscr{S}^\prime(\X)$ the integral operator having the distribution kernel $T$, given by the formula
\beq
\langle u,\big(\Int(T)\big)v\rangle_\mathcal{H}\,:=\,\big(T,\overline{u}\otimes v\big),\qquad\forall(u,v)\in\mathscr{S}(\X)\times\mathscr{S}(\X)\,.
\eeq
For any tempered distribution $F\in\mathscr{S}^\prime(\Xi)$ we denote by $K_F\in\mathscr{S}^\prime(\X\times\X)$ the integral kernel of its Weyl quantization, i.e. $\mathfrak{Op}(F)=\Int(K_F)$. Let us recall that there exists a  linear bijection
$\mathfrak{W}:\mathscr{S}^\prime(\Xi)\rightarrow\mathscr{S}^\prime(\X\times\X)$
defined by 
 \beq\label{iulie3}
\big(\mathfrak{W}F\big)(x,y)\ :=\
(2\pi)^{-2}\int_{\X^*}e^{i<\xi,x-y>}F\big(\frac{x+y}{2},\xi\big)\,d\xi\,,
\eeq
such that
$$ \mathfrak{Op}(F)=\mathcal{I}\text{\sf
nt}(\mathfrak{W}F)\,.$$
In the magnetic calculus, we have the equality
\beq\label{magn-quant}
\mathfrak{Op}^A(F)\ =\ \mathcal{I}\text{\sf
nt}({\Lambda}^A\,\mathfrak{W}F)\,.
\eeq

\section{The magnetic quasi-band projections}
\label{S.5}

Following step by step the arguments in Section 3 of \cite{CHP} we can associate with  the quasi-band orthogonal projection $\pi\in\mathcal{L}(\mathcal{H})$ some \textit{magnetic versions} of it: $\pi^{\epsilon}$ and $\pi^{\epsilon,\kappa}$.
Starting from the principal quasi Wannier function $\psi_0$ introduced in Definition \ref{pseudo-W-f} and the magnetic field introduced in \eqref{Bek} we define the following objects:
\begin{definition}\label{m-q-W-def}~
\begin{enumerate}
\item $\mathring{\phi}^\epsilon_\gamma(x):=\Lambda^\epsilon(x,\gamma)\psi_0(x-\gamma)$\,, \,$\mathbb{G}^\epsilon_{\alpha\beta}:=\langle\mathring{\phi}^
\epsilon_\alpha\,,\,\mathring{\phi}^\epsilon_\beta\rangle_{\mathcal{H}}$\,,\,  $\mathbb{F}^\epsilon:=\big(\mathbb{G}^\epsilon\big)^{-1/2}\in
\mathcal{L}\big(\ell^2(\Gamma)\big)$,
\item  $\psi^\epsilon_0(x):=\underset{\alpha\in\Gamma}{\sum}\mathbb{F}^\epsilon_{\alpha\,0}\, \Omega^\epsilon
(\alpha,0,x)\psi_0(x-\alpha)$\,, \,$\phi^\epsilon_\gamma(x):=
\Lambda^\epsilon(x,\gamma)\, \psi^\epsilon_0(x-\gamma)$.
\item $\pi^\epsilon :=\sum_{\gamma\in \Gamma} |\phi^\epsilon_\gamma\rangle\langle \phi^\epsilon_\gamma| =(\pi^\epsilon)^2=(\pi^\epsilon)^*$\,,\,
\item $\mathring{\phi}^{\epsilon,\kappa}_\gamma(x):=\Lambda^{\epsilon,\kappa}(x,\gamma)\psi^\epsilon_0(x-\gamma)$\,, \,$\mathbb{G}^{\epsilon,\kappa}_{\alpha\beta}:=\langle\mathring{\phi}^{\epsilon,\kappa}_\alpha\,,\,\mathring{\phi}^{\epsilon,\kappa}_\beta\rangle_{\mathcal{H}}$\,, \, $\mathbb{F}^{\epsilon,\kappa}:=\big(\mathbb{G}^{\epsilon,\kappa}\big)^{-1/2}\in
\mathcal{L}\big(\ell^2(\Gamma)\big)$,
\item $\phi^{\epsilon,\kappa}_\gamma:=\underset{\alpha\in\Gamma}{\sum}
\mathbb{F}^{\epsilon,\kappa}_{\alpha\gamma}\mathring{\phi}^{\epsilon,\kappa}_\alpha$\,,\quad $\pi^{\epsilon,\kappa}:=\sum_{\gamma\in \Gamma} |\phi^{\epsilon,\kappa}_\gamma\rangle\langle \phi^{\epsilon,\kappa}_\gamma| =(\pi^{\epsilon,\kappa})^2=(\pi^{\epsilon,\kappa})^*$.
\end{enumerate}
\end{definition}
Using the notation introduced in \eqref{Aek}-\eqref{sh-not-magn-2} we notice that 
\beq\label{Lambda-dec}
\Lambda^{\epsilon,\kappa}(x,y)=\widetilde{\Lambda}^{\epsilon,\kappa}(x,
y)\Lambda^\epsilon(x,y)
\eeq
so that
\beq \label{phase-fact-dec}
\overset{\circ}{\phi_\gamma^{\epsilon,\kappa}}(x)\,=\,\widetilde{\Lambda}^{
\epsilon,\kappa}(x,\gamma)\phi^\epsilon_\gamma(x)
\eeq
If we simply replace the Wannier function $\phi_0$ in Section 3 of \cite{CHP} with the quasi Wannier function $\psi_0$, all the results of Section 3 of \cite{CHP} remain true due to the fact that the only properties of $\phi_0$ that are used are its smoothness and rapid decay together with all its derivatives, and these facts are true for the quasi Wannier function too. Thus we obtain the following statement, extending the results in Section 3 of \cite{CHP}:

\begin{proposition}\label{Prop-m-q-W}~
There exists $\epsilon_0>0$ such that:
\begin{enumerate}
\item  $\mathbb{F}^\epsilon$ belongs to $\mathcal{L}\big(\ell^2(\Gamma)\big)\bigcap
\mathcal{L}\big(\ell^\infty(\Gamma)\big)$  for any $\epsilon\in[0,\epsilon_0]\,$, and satisfies:
$$\forall m\in\mathbb{N},\ \exists C_m \mbox{ s.t. } <\alpha-\beta>^m\left|\mathbb{F}^\epsilon_{\alpha\beta}-\delta_{\alpha\beta}\right|\leq\,C_m\,\epsilon\,,\ \forall(\alpha,\beta)\in\Gamma^{2}\,,\, \forall \epsilon\in[0,\epsilon_0]\,, 
$$ and there exists a rapidly decaying function $\mathbf{F}^\epsilon:\Gamma\mapsto \mathbb{C}$ such that
$$
\mathbb{F}^\epsilon_{\alpha\beta}\ =\ \Lambda^\epsilon(\alpha,\beta)\mathbf{F}^\epsilon(\alpha-\beta)\,.
$$
\item For any $m\in\mathbb{N}$ and any $a\in\mathbb{N}^2\,$, there exists $C_{m,a}>0$ such that
$$
\underset{x\in\X}{\sup}<x>^m\left|\big(\partial^a\psi^\epsilon_0\big)(x)-\big(\partial^a\psi_0\big)(x)\right|\ \leq\ C_{m,a}\epsilon\,, \qquad\forall\epsilon\in[0,\epsilon_0]\,.
$$
\item For any $m\in\mathbb{N}\,$, there exists $C_m>0$ such that
\beq
\underset{(\alpha,\beta)\in\Gamma^2}{\sup}
<\alpha-\beta>^m\left|\mathbb{F}^{\epsilon,\kappa}_{\alpha,\beta}-\delta_{\alpha\beta}
\right|\leq C_m\, \kappa\epsilon\,,\
\forall(\epsilon,\kappa)\in[0,\epsilon_0]\times[0,1]\,.
\eeq
\end{enumerate}
\end{proposition}

Replacing $\psi_0$ with $\psi^\epsilon_0$ in the proof of Proposition \ref{bd-Ham0-band} we obtain the following result:

\begin{proposition}\label{bd-Hameps-band}
There exists $\epsilon_0>0$ such that for any $(\epsilon,\kappa)\in[0,\epsilon_0]\times[0,1]$ we have 
 $${\pi}^{\epsilon,\kappa}\mathcal{H}\subset\mathcal{D}(H^{\epsilon,\kappa})\,,$$
  while 
$H^{\epsilon,\kappa}{\pi}^{\epsilon,\kappa}\in\mathcal L(\mathcal{H})$ and
${\pi}^{\epsilon,\kappa} H^{\epsilon,\kappa}$ has a bounded closure.
\end{proposition}

We can also construct the $\Gamma$-periodic symbol $p_\epsilon\in
S^{-\infty}(\Xi)$ such that $\pi^\epsilon=\mathfrak{Op}^\epsilon(p_\epsilon)$ and the symbol $p_{\epsilon,\kappa}\in
S^{-\infty}(\Xi)$ such that ${\pi}^{\epsilon,\kappa}=\mathfrak{Op}^{\epsilon,\kappa}(p_{\epsilon,\kappa})$ and the same proof as in \cite{CHP} gives  the following result:
\begin{proposition}\label{peps-p}
There exists $\epsilon_0>0$ such that for any seminorm  $\nu$ on
$S^{-\infty}(\Xi)$, there exists $C_\nu>0$
such that 
$$
\nu(p^\epsilon-p)\leq\,C_\nu\,\epsilon 
\mbox{ and } \,
\nu(p^{\epsilon,\kappa}-p^\epsilon)\leq\,C_\nu\,\kappa\epsilon\,,\,
\forall(\epsilon,\kappa)\in[0,\epsilon_0]\times[0,1]\,.
$$
\end{proposition}

Note that in this case the commutator
$\big[H^{\epsilon,\kappa},{\pi}^{\epsilon,\kappa}\big]$ is no longer
small (due to the arbitrary deformation which was made in constructing the quasi Wannier function).

\section{Checking the conditions of the abstract reduction procedure for the magnetic operators}
\label{S.6}

We want to apply  Propositions \ref{bernard7'} and \ref{C-FS} for magnetic pseudodifferential operators and to control the behavior of their symbols. The following abstract result concerns the symbol of the "reduced resolvent" $R_\bot(E)$ which is introduced after \eqref{eq:lb}. In our applications below we shall replace $H$ with $H-E\bb1$ for some $E<\inf\sigma(H)$.

Suppose given some real elliptic symbol $h\in S^m_1(\Xi)_{\text{\sf ell}}$ for
some $m>0$ and  consider the magnetic pseudodifferential calculus,
denoted by $\mathfrak{Op}^A$, associated with  a smooth polynomially bounded vector potential $A$ such the  magnetic field $B=\curl A$ is of class $BC^\infty(\X)$. Let $$H= \mathfrak{Op}^A(h)$$ and suppose also given an orthogonal
projection $\Pi\equiv\mathfrak{Op}^A(p)$ with $p\in S^{-\infty}(\Xi)$ and such
that $\Pi\mathcal{H}\subset\mathcal{D}(H)$. Let us also
introduce  $$H^\bot:=\Pi^\bot H\Pi^\bot \mbox{ with }
\Pi^\bot:=\bb1-\Pi\,.
$$

\begin{proposition}\label{red-mPsD}
With the above notation, if $H^\bot$ is invertible on
$\Pi^\bot\mathcal{H}$ with bounded inverse
$R\in\mathcal L(\Pi^\bot\mathcal{H})$, then there exists
$r\in S^{-m}(\Xi)$ such that $\Pi^\bot R\Pi^\bot=\mathfrak{Op}^A(r)$.
\end{proposition}

\begin{proof}
We notice that 
\beq
\Pi^\bot R\Pi^\bot=\Pi^\bot\big(\Pi\,+\,\Pi^\bot H\Pi^\bot\big)^{-1}\Pi^\bot
\eeq
and use point (3) of Proposition \ref{IMP}.
\end{proof}

Let us verify that the pair of operators $(H^{\epsilon,\kappa},\pi^{\epsilon,\kappa})$ satisfies Hypothesis \ref{Prop-FS}.

\begin{proposition}\label{FS-Hek}
There exist $\epsilon_0>0$ and $b_0\in(0,\widetilde{b})$ (with $\widetilde{b}>0$ introduced in Lemma \ref{newlemma}) such that for any $(\epsilon,\kappa)\in[0,\epsilon_0]\times[0,1]$ the pair of operators  $(H^{\epsilon,\kappa},\pi^{\epsilon,\kappa})$ verify the following properties:
\begin{enumerate}
 \item $H^{\epsilon,\kappa}\pi^{\epsilon,\kappa}$ is 
bounded on $\mathcal{H}$, uniformly in $\epsilon$ and $\kappa\,$.
\item 
$(\bb1-\pi^{\epsilon,\kappa})H^{\epsilon,\kappa}(\bb1-\pi^{\epsilon,\kappa}
)\geq b_0\,(\bb1-\pi^{\epsilon,\kappa})\,$.
\end{enumerate}
\end{proposition}

\begin{proof}
The first statement is a direct consequence of Proposition \ref{bd-Hameps-band} and we focus on the second statement. 

Let $b^\prime<b<\widetilde{b}$ with $\widetilde{b}>0$ introduced in Lemma \ref{newlemma} so that the closure of $\Sigma_{b^\prime}$ is included in $\Sigma_{b}$.  We choose a function $g\in C^\infty_0(\Sigma_{b})$ such that $0\leq
g(\theta)\leq1$ and $g(\theta)=1$ on $\Sigma_{b^\prime}$. Then we define successively (using \eqref{hatpi}):
\beq\label{def-K}
\hat{K}^0(\theta):=\hat{H}^0(\theta)+g(\theta)\hat{\pi}(\theta)\,,\qquad
K^0:=\mathscr{U}_\Gamma^{-1}\left(\int_{E^*}^\oplus\hat{K}^0(\theta)\,
d\theta\right)\mathscr{U}_\Gamma\,,
\eeq
\beq\label{def-W}
W^0:=K^0-H^0,\quad w:=\mathfrak{S}(W^0),\quad
W^{\epsilon,\kappa}:=\mathfrak{Op}^{\epsilon,\kappa}(w),\quad K^{\epsilon,\kappa}:=H^{\epsilon,\kappa} +W^{\epsilon,\kappa}\,.
\eeq
By construction,  we have $K^0\geq b^\prime\bb1$, which implies:
\beq\label{l-bound-0}
\pi^\bot H^0\pi^\bot=\pi^\bot K^0\pi^\bot\geq b^\prime\pi^\bot.
\eeq
Using the regularity of the spectrum (see Corollary 1.6 in
\cite{CP-2}), we conclude that for $\epsilon_0>0$ small enough,
$K^{\epsilon,\kappa}\geq(3/4)b^\prime$ for every
$\epsilon\in[0,\epsilon_0]$ and $0\leq \kappa\leq 1$, and this  implies:
\beq\label{l-bound}
\big(\bb1-\pi^{\epsilon,\kappa}\big)K^{\epsilon,\kappa}\big(\bb1-\pi^{\epsilon,\kappa}\big)\geq(3/4)b^\prime\big(\bb1-\pi^{\epsilon,\kappa}\big).
\eeq
By construction, we have that $W^0=\pi W^0\pi$. We show that their magnetic counterparts are also close in norm. 
We write:
$$
\pi^{\epsilon,\kappa}W^{\epsilon,\kappa}\pi^{\epsilon,\kappa}-W^{\epsilon,\kappa}=\mathfrak{Op}^{\epsilon,\kappa}
\left(p^{\epsilon,\kappa}\,\sharp^{\epsilon,\kappa}\, 
 w\,\sharp^{\epsilon,\kappa}\, 
p^{\epsilon,\kappa}-w\right),
$$
while 
\begin{align*}
&p^{\epsilon,\kappa}\,\sharp^{\epsilon,\kappa}\, 
 w\,\sharp^{\epsilon,\kappa}\, 
p^{\epsilon,\kappa}-p\,\sharp\, w\,\sharp\,
p \\
&=\big(p^{\epsilon,\kappa}\,\sharp^{\epsilon,\kappa}\, 
 w-p^{\epsilon,\kappa}\,\sharp\,
w\big)\,\sharp^{\epsilon,\kappa}\, 
p^{\epsilon,\kappa}+p^{\epsilon,\kappa}
\,\sharp\,\big(w\,\sharp\,^{\epsilon, \kappa}
p^{\epsilon,\kappa}-w\,\sharp\, p\big)+\big(p^{\epsilon,\kappa}-p\big)\,\sharp\,
w\,\sharp\, p
\end{align*}
is a symbol of class $S^0_1(\Xi)\subset\mathfrak{C}^{\epsilon,\kappa}(\Xi)$ having the norm in $\mathfrak{C}^{\epsilon,\kappa}(\Xi)$ (i.e. the operator norm of its magnetic quantization $\mathfrak{Op}^{\epsilon,\kappa}$) of order $\epsilon$ and thus
there exist $C$ and $\epsilon_0$ such that 
\beq\label{horiac9}
\left\|\pi^{\epsilon,\kappa}W^{\epsilon,\kappa}\pi^{\epsilon,\kappa}-W^{\epsilon,\kappa}\right\|\leq
C\, \epsilon\,,
\eeq
for any $(\epsilon,\kappa)\in[0,\epsilon_0]\times[0,1]$. 

Finally, writing 
$$
(\bb1-\pi^{\epsilon,\kappa})H^{\epsilon,\kappa}(\bb1-\pi^{\epsilon,\kappa})=(\bb1-\pi^{\epsilon,\kappa})
K^{\epsilon,\kappa}(\bb1-\pi^{\epsilon,\kappa})-(\bb1-\pi^{\epsilon,\kappa})
(W^{\epsilon,\kappa}-\pi^{\epsilon,\kappa}  
W^{\epsilon,\kappa}\pi^{\epsilon,\kappa})(\bb1-\pi^{\epsilon,\kappa})
$$
and using \eqref{l-bound} and \eqref{horiac9} we see that if $\epsilon$ is small enough then we can choose $b_0=b^\prime/2\,$. This achieves the proof.
\end{proof}
Thus Hypothesis \ref{Prop-FS} is satisfied when $H$ and $\Pi$ are replaced by $H^{\epsilon,\kappa}$ and $\pi^{\epsilon,\kappa}$ respectively, and $\beta=b'/2\,$. Hence Proposition \ref{C-FS} can be applied and fixing $N>1$ and working with $\epsilon \geq 0$ such that $\epsilon N<b'/2\,$, we deduce that $H^{\epsilon,\kappa}$ has $N$ gaps of order $\epsilon$ in its lower spectrum if the operator $\widetilde{H}^{\epsilon,\kappa}$ associated with the pair
$(H^{\epsilon,\kappa},\pi^{\epsilon,\kappa})$ as in \eqref{dress} and \eqref{horiac6} also has $N$ gaps of order $\epsilon$ in its lower spectrum. We shall use the notations $Y^{\epsilon,\kappa}\in\mathcal{L}(\mathcal{H})$ and $\widetilde{H}^{\epsilon,\kappa}\in\mathcal{L}(\pi^{\epsilon,\kappa}\mathcal{H})$ for the operators defined as in \eqref{dress} and \eqref{horiac6} for $H$ replaced by $H^{\epsilon,\kappa}$ and $\Pi$ replaced by $\pi^{\epsilon,\kappa}$.

\section{The one band effective Hamiltonian}
\label{modif-B-function}

The conclusion of  Section \ref{S.6}  implies that in order to finish the proof of our main result in Theorem \ref{mainTh} we have to analyze the bottom of the spectrum of the operator $\pi^{\epsilon,\kappa}\widetilde{H}^{\epsilon,\kappa}\pi^{\epsilon,\kappa}$.
We follow essentially  the arguments and ideas from  Subsection 3.3 in \cite{CHP}, but due to the fact that in the present situation the calculus is much more involved, we prefer to present rather complete proofs. 
\subsection{Preliminaries}
As we shall start with the situation with constant magnetic field, let us recall the following results concerning the \textit{Zak magnetic translations} \cite{Z} (see also  \cite{Ne-RMP,HS,CHN}).

\begin{proposition}\label{Zak-transl}
In the case of a constant magnetic field of the form $B_\epsilon=\epsilon B_0$ the family of unitary operators
\beq
\mathcal{T}^\epsilon_\gamma:=\Lambda^\epsilon(\cdot,\gamma)\tau_{\gamma}\,,\qquad \gamma\in\Gamma
\eeq
satisfies the following properties:
\begin{enumerate}
\item $\mathcal{T}^\epsilon_\alpha\mathcal{T}^\epsilon_\beta=
\Lambda^\epsilon(\beta,\alpha)\mathcal{T}^\epsilon_{\alpha+\beta}$\,.
\item The operator $\mathfrak{Op}^\epsilon(F)$ commutes with all the $\{\mathcal{T}^\epsilon_\gamma\}_{\gamma\in\Gamma}$ if and only if $F\in\mathscr{S}^\prime(\Xi)$ is $\Gamma$-periodic with respect to the variable in $\X$.
\end{enumerate}
\end{proposition}
On the other hand, looking at the definition of the magnetic Moyal product and using the second point of the above proposition one can prove the following statement.
\begin{proposition}\label{Gamma-inv}
For a constant magnetic field $B_\epsilon:=\epsilon B_0$ the following statements hold true:
\begin{enumerate}
\item The subspace $\mathfrak M^{B_\epsilon}_\Gamma$ of the $\Gamma$-periodic distributions of  $\mathfrak M^{B_\epsilon}$ is  a subalgebra.
\item If $F \in  \mathfrak M^{B_\epsilon}_\Gamma  $ is invertible  in  $\mathfrak M^{B_\epsilon}$, its inverse is  in $\mathfrak M^{B_\epsilon}_\Gamma$.
\end{enumerate}
\end{proposition}

Having in mind Proposition \ref{MP} and with the shorthand notation (see \eqref{defsymb})
\beq
\mathfrak{S}^{\epsilon,\kappa}:=\mathfrak{S}^{A^{\epsilon,\kappa}}\,,\qquad
\mathfrak{S}^{\epsilon}:=\mathfrak{S}^{A^{\epsilon}}\,,
\eeq
we introduce the following symbols:

\begin{definition}\label{rem-Hpi}~
\begin{itemize}
 \item
$h_\circ^{\epsilon,\kappa}:=\mathfrak{S}^{\epsilon,\kappa}\big({\pi}^{\epsilon,\kappa}H^{\epsilon,\kappa}{\pi}^{\epsilon,
\kappa}\big);\quad
h_\circ^{\epsilon}:=\mathfrak{S}^{\epsilon}\big({\pi}^{\epsilon}H^{\epsilon}{\pi}^{\epsilon}\big)$;
 \item
$h_\bot^{\epsilon,\kappa}:= \mathfrak{S}^{\epsilon,\kappa}\big(\big(\bb1-{\pi}^{
\epsilon,\kappa}\big)H^{\epsilon,\kappa}\big(\bb1-{\pi}^{
\epsilon,\kappa}\big)\big);\quad
h_\bot^{\epsilon}:=\mathfrak{S}^{\epsilon}\big(\big(\bb1-{\pi}^{
\epsilon}\big)H^{\epsilon}\big(\bb1-{\pi}^{
\epsilon}\big)\big)$;
 \item
$h_\bullet^{\epsilon,\kappa}:=\mathfrak{S}^{\epsilon,\kappa}\big({\pi}^{\epsilon,\kappa}H^{\epsilon,\kappa}\big(\bb1-{\pi}^{
\epsilon,\kappa}\big)\big);\quad h_\bullet^{\epsilon}:=
\mathfrak{S}^{\epsilon}\big({\pi}^{\epsilon}H^{\epsilon}
\big(\bb1-{\pi}^{\epsilon}\big)\big)$.
\item $r_{\epsilon,\kappa}:=\mathfrak{S}^{\epsilon,\kappa}\big((\bb1-{\pi}^{\epsilon,
\kappa})R^{\epsilon,\kappa}_\bot(\bb1-{\pi}^{\epsilon,
\kappa})\big);\quad  r_{\epsilon}:=\mathfrak{S}^{\epsilon}\big(
(\bb1-{\pi}^{\epsilon})R^{\epsilon}_\bot(\bb1-{\pi}^{
\epsilon})\big)$;
\item 
$\mathfrak{y}^{\epsilon,\kappa}:=\mathfrak{S}^{\epsilon,\kappa}\big(\pi^{\epsilon,\kappa}\big(Y^{\epsilon,\kappa})^{-1/2}\pi^{\epsilon,\kappa}\big)\,;\quad\mathfrak{y}^{\epsilon}:=\mathfrak{S}^{\epsilon}
\big(\pi^{\epsilon}\big(Y^{\epsilon}\big)^{-1/2}\pi^{\epsilon}
\big)\,$;
\item $\mathfrak{z}^{\epsilon,\kappa}:=\mathfrak{y}^{\epsilon,\kappa}-p^{\epsilon,\kappa};\quad \quad
\mathfrak{z}^{\epsilon}:=\mathfrak{y}^{\epsilon}-p^{\epsilon}\,$
\end{itemize}
\end{definition}

\begin{proposition}
We have the following properties:
\begin{enumerate}
\item $\{r_{\epsilon,\kappa},r_{\epsilon}\} \subset S^{-2}(\Xi)\,$, 
\item $\big\{\mathfrak{y}^{\epsilon,\kappa},\mathfrak{z}^{\epsilon,\kappa},\mathfrak{y}^{\epsilon},\mathfrak{z}^{\epsilon}\big\}\subset S^{-\infty}(\Xi)\,$.
\end{enumerate}
\end{proposition}
\begin{proof}
For the first two conclusions we just use Proposition \ref{red-mPsD} with $m=2$. For the third conclusion we notice that $\pi^{\epsilon,\kappa}$ and $\pi^\epsilon$ have symbols of class $S^{-\infty}(\Xi)$.
\end{proof}

\subsection{Main approximation }\label{SS-approx-matrix}

\subsubsection{ Reduction to the constant field $\epsilon B_0$}
 The next proposition states that $ \widetilde{H}^{\epsilon,\kappa} $ is close in operator norm to a magnetic matrix in which the non-constant magnetic field appears only through the phase $\widetilde{\Lambda}^
{\epsilon,\kappa}$. 
\begin{proposition}\label{prelim-res}
With the definitions and notations  of this section, the following approximation holds, for any $(\alpha,\beta) \in \Gamma\times \Gamma$, 
\begin{equation}\label{eq:7.2.1.a}
\begin{array}{l}
\left\langle\phi^{\epsilon,\kappa}_\alpha\,,\,
\widetilde{H}
^{\epsilon,\kappa}
\phi^{\epsilon,\kappa}_\beta
\right\rangle_{\mathcal{H}}=\widetilde{\Lambda}^
{\epsilon,\kappa}(\alpha,\beta)
\left\langle\phi^{\epsilon}_\alpha\,,\,\mathfrak{Op}^\epsilon(\mathfrak h^\epsilon)
\phi^{\epsilon}_\beta
\right\rangle_{\mathcal{H}}+\kappa \epsilon \, 
\mathscr{O} 
(|\alpha-\beta|^{-\infty})\,,
\end{array}
\end{equation}
with 
$$
\mathfrak h_\epsilon:= h_\circ^\epsilon+h_\circ^{\epsilon}\,\sharp\,^{\epsilon}
\mathfrak{z}^{\epsilon}
+\mathfrak{z}^{\epsilon}\,\sharp\,^{
\epsilon}h_\circ^{\epsilon}\,\sharp\,^{\epsilon}\mathfrak{y}^{\epsilon}+\mathfrak{y}^{\epsilon}\,\sharp\,^{\epsilon}h_\bullet^{\epsilon}
\,\sharp\,^{\epsilon}r_{\epsilon}\,\sharp\,^{\epsilon}
h_\bullet^{\epsilon}\,\sharp\,^{\epsilon}\mathfrak{y}^{\epsilon}\,.
$$
Moreover
\begin{equation}\label{eq:7.2.1.b}
\langle\phi^{\epsilon}_\alpha\,,\,\mathfrak{Op}^\epsilon(\mathfrak h^\epsilon)
\phi^{\epsilon}_\beta\rangle_{\mathcal{H}}=\Lambda^\epsilon(\alpha,\beta)\langle\phi^{\epsilon}_{\alpha-\beta}\,,\,\mathfrak{Op}^\epsilon(\mathfrak h^\epsilon)
\phi^{\epsilon}_0 \rangle_{\mathcal{H}}\,.
\end{equation}
\end{proposition}

\begin{proof}
Having in mind the notation of  Section \ref{S.5} (see Definition \ref{m-q-W-def}, item 5), we get:
\beq\label{magn-matrix'}
\left\langle\phi^{\epsilon,\kappa}_\alpha\,,\,\pi^{\epsilon,\kappa}
\widetilde{H}
^{\epsilon,\kappa}\pi^{\epsilon,\kappa}
\phi^{\epsilon,\kappa}_\beta
\right\rangle_{\mathcal{H}} =  \Sigma_1(\alpha,\beta) + \Sigma_2(\alpha,\beta)\,,
\eeq
where, using \eqref{horiac6} we can write
\begin{equation}\label{magn-matrix-1}
\Sigma_1(\alpha,\beta):= \underset{\alpha'\in\Gamma}{
\sum}\underset{\beta'\in\Gamma}{\sum}\overline{\mathbb{F}}^{\epsilon,\kappa}_{\alpha'\alpha}
\mathbb{F}^{\epsilon,\kappa}_{\beta'\beta}
\left\langle\big(Y^{\epsilon,\kappa}\big)^{-\frac{1}{2}}\Lambda^{\epsilon,\kappa}(\cdot,
\alpha')\tau_{\alpha'}\psi^\epsilon_0\,,\,H^{\epsilon,\kappa}
\big(Y^{\epsilon,\kappa}\big)^{-\frac{1}{2}}\Lambda^{\epsilon,\kappa}(\cdot,\beta')\tau_
{\beta'}\psi^\epsilon_0\right\rangle_{\mathcal{H}} 
\end{equation}
and 
\begin{equation}\label{magn-matrix-2}
\Sigma_2(\alpha,\beta):=
 \underset{\alpha'\in\Gamma}{
\sum}\underset{\beta'\in\Gamma}{\sum}\overline{\mathbb{F}}^{\epsilon,\kappa}_{\alpha'\alpha}
\mathbb{F}^{\epsilon,\kappa}_{\beta'\beta}
\left\langle\big(Y^{\epsilon,\kappa}\big)^{-\frac{1}{2}}\Lambda^{\epsilon,\kappa}(\cdot,
\alpha')\tau_{\alpha'}\psi^\epsilon_0\,,\,
K^{\epsilon,\kappa}
\big(Y^{\epsilon,\kappa}\big)^{-\frac{1}{2}}\Lambda^{\epsilon,\kappa}(\cdot,\beta')\tau_
{\beta'}\psi^\epsilon_0\right\rangle_{\mathcal{H}}
\end{equation}
with $$K^{\epsilon,\kappa}:=\pi^{\epsilon,\kappa}H^{\epsilon,\kappa}R_\bot^{\epsilon,\kappa}(0)H^{\epsilon,\kappa}\pi^{\epsilon,\kappa}\in\mathcal{L}(\pi^{\epsilon,\kappa}\mathcal{H})\,.$$
Using  Formula \eqref{phase-fact-dec} and Proposition \ref{Prop-m-q-W} (point 3), we will  express the series appearing in  $ \Sigma_1(\alpha,\beta)$ and $\Sigma_2(\alpha,\beta)$  in terms of elements associated with  $\pi^\epsilon\mathcal{H}$ in the constant magnetic field $\epsilon B_0$.

In Subsection \ref{SS-sigma-1} we will prove the following estimation contained in Lemma \ref{Concl-1}:
$$\begin{array}{ll}
\Sigma_1(\alpha,\beta)&
=\widetilde{\Lambda}^{\epsilon,\kappa}(\alpha,\beta)
\left\langle\phi^\epsilon_{
\alpha}\,,\,\left[H^{\epsilon}+\mathfrak{Op}^{\epsilon}\big(h_\circ^{\epsilon,
\kappa}\,\sharp^{\epsilon,\kappa}\, 
\mathfrak{z}^{\epsilon,\kappa}
\big)+\mathfrak{Op}^{\epsilon}\big(\mathfrak{z}^{\epsilon,\kappa}\,\sharp\,^{
\epsilon,\kappa}h_\circ^{\epsilon,\kappa}\,\sharp^{\epsilon,\kappa}\, 
\mathfrak{y}^{\epsilon,\kappa}\big)\right]\phi^\epsilon_{\beta}\right\rangle_{\mathcal{H
}}\\ & \qquad + \kappa\epsilon\, \mathscr{O}(|\alpha-\beta|^{-\infty}) \,.
\end{array}
$$
For the term $\Sigma_2(\alpha,\beta)$ we will prove in Subsection \ref{SS-sigma-2} (Lemma \ref{Concl-2}) that
$$
\Sigma_2 (\alpha,\beta)
=\widetilde{\Lambda}^{\epsilon,\kappa}(\alpha,\beta)
\left\langle\phi^\epsilon_{
\alpha}\,,\,\mathfrak{Op}^{\epsilon}\big(
\mathfrak{y}^{\epsilon,\kappa}\,\sharp^{\epsilon,\kappa}\, 
h_\bullet^{\epsilon,
\kappa}\,\sharp^{\epsilon,\kappa}\, 
r_{\epsilon,\kappa}\,\sharp^{\epsilon,\kappa}\, 
h_\bullet^{\epsilon,\kappa}\,\sharp^{\epsilon,\kappa}\, 
\mathfrak{y}^{\epsilon,
\kappa}\big)\phi^\epsilon_{\beta}\right\rangle_{\mathcal{H
}}+ \kappa \epsilon \, \mathscr{O}(|\alpha-\beta|^{-\infty})\,.
$$

We continue by noticing that
Proposition \ref{peps-p} and the properties of the magnetic composition of
symbols (see Proposition B.12 in \cite{CHP}) imply that for any seminorm
$\nu$ on $S^{-\infty}(\Xi)$ there exists $C_\nu>0$ such that
\beq\label{est-h}
\nu(h_\circ^{\epsilon,\kappa}-h_\circ^{\epsilon})\leq C_\nu\,\kappa\epsilon\,,\quad
\nu(h_\bullet^{\epsilon,\kappa}-h_\bullet^{\epsilon})\leq
C_\nu\,\kappa\epsilon\,,\quad\nu(h_\bot^{\epsilon,\kappa}-h_\bot^{
\epsilon})\leq C_\nu\,\kappa\epsilon\,.
\eeq

Then let us consider the difference $r_{\epsilon,\kappa}-r_{\epsilon}\in
S^{-2}(\Xi)$ and compute
\beq
\big[h_\bot^{\epsilon,\kappa}\,\sharp^{\epsilon,\kappa}\, 
\big(r_{\epsilon,
\kappa } -r_{\epsilon}\big)\big]\,\sharp\,^{\epsilon}h_\bot^{\epsilon}=(1-
p^{\epsilon,\kappa})\,\sharp\,^{\epsilon}h_\bot^{\epsilon}-
\big(h_\bot^{\epsilon,\kappa}\,\sharp^{\epsilon,\kappa}\, 
r_\epsilon\big)\,\sharp\,^{\epsilon}h_\bot^{\epsilon}.
\eeq
Using Propositions B.12 and B.14 in \cite{CHP} and the estimates in Proposition
\ref{peps-p} we conclude that
\beq
\big(h_\bot^{\epsilon,\kappa}\,\sharp^{\epsilon,\kappa}\, 
r_\epsilon\big)\,\sharp\,^{\epsilon}h_\bot^{\epsilon}  =\big(h_\bot^{
\epsilon,\kappa}\,\sharp\,^{\epsilon}
r_\epsilon\big)\,\sharp\,^{\epsilon}h_\bot^{\epsilon}+\kappa\epsilon\mathfrak
{r}_{1,\epsilon,\kappa} =h_\bot^{\epsilon,\kappa}\,\sharp\,^{\epsilon}(1-p^{
\epsilon})+\kappa\epsilon\, \mathfrak{r}_{1,\epsilon,\kappa},
\eeq
and 
\beq
\begin{array}{ll}
\big[h_\bot^{\epsilon,\kappa}\,\sharp^{\epsilon,\kappa}\, 
\big(r_{\epsilon,
\kappa}-r_{\epsilon}\big)\big]\,\sharp\,^{\epsilon}h_\bot^{\epsilon}& =(1-
p^{\epsilon,\kappa})\,\sharp\,^{\epsilon}h_\bot^{\epsilon}-h_\bot^{
\epsilon,\kappa}\,\sharp\,^{\epsilon}(1-p^{
\epsilon})-\kappa\epsilon\, \mathfrak{r}_{1,\epsilon,\kappa}\\ & 
=h_\bot^{\epsilon}-h_\bot^{\epsilon,\kappa}
+\kappa\epsilon\, \mathfrak{r}_{2,\epsilon,\kappa}\\ & =(1-p^{\epsilon})\,\sharp\,^{\epsilon
}h\,\sharp\,^{\epsilon}(1-p^{\epsilon})-(1-p^{\epsilon,\kappa})\,\sharp\,^{\epsilon,
\kappa}h\,\sharp^{\epsilon,\kappa}\, 
(1-p^{\epsilon,\kappa})\\ & \qquad -\kappa\epsilon\, \mathfrak{
r}_{3,\epsilon,\kappa}\,,
\end{array}
\eeq
with $\mathfrak{r}_{1,\epsilon,\kappa}\,,$ $\mathfrak{r}_{2,\epsilon,\kappa}$ and $\mathfrak{r}_{3,\epsilon,\kappa}$
in $S^0(\Xi)$ uniformly for $(\epsilon,\kappa)\in[0,\epsilon_0]\times[0,1]$.\\
Using Proposition \ref{peps-p} several times we obtain that the difference $r_{\epsilon,\kappa}-r_{\epsilon}$ belongs to $S^{-4}(\Xi)$
and for any seminorm $\nu$ on  $S^{-4}(\Xi)$
there exists $C_\nu>0$ such that
\beq\label{est-r}
\nu\big(r_{\epsilon,\kappa}-r_{\epsilon}\big)\leq C_\nu\, \kappa\epsilon\,.
\eeq

From Proposition \ref{peps-p} and the magnetic version of
Calderon-Vaillancourt Theorem (see Theorem~3.1 in \cite{IMP1}) we deduce that
\beq
\left\|\mathfrak{Op}^\epsilon\big(p^{\epsilon,\kappa}
-p^\epsilon\big)\right\|\leq C\, \kappa\epsilon\,.
\eeq
Meanwhile, from Proposition B.14 in \cite{CHP},  we also
deduce that for any symbols $f$ and $g$ in  $S^0(\Xi)$ there exists $C>0$ such that
\beq
\left\|\mathfrak{Op}^\epsilon\big(f\,\sharp^{\epsilon,\kappa}\, 
g-f\,\sharp\,^\epsilon g\big)\right\|\leq C\,\kappa\epsilon\,.
\eeq

Putting all these results together, using \eqref{est-h} and \eqref{est-r} in order to control the difference
$\mathfrak{z}^{\epsilon,\kappa}-\mathfrak{z}^\epsilon$ and the usual result
concerning the estimate 
 of the difference of the two magnetic products
$\,\sharp^{\epsilon,\kappa}\, 
$ and $\,\sharp\,^\epsilon$ (see Proposition B.14 in \cite{CHP}) we finally obtain:
\beq\label{S-1-magn-matrix}
\Sigma_1(\alpha,\beta) \,=\,\widetilde{\Lambda}^{\epsilon,\kappa}(\alpha,\beta)
\left\langle\phi^\epsilon_{
\alpha}\,,\, \mathfrak{Op}^{\epsilon} \big(h_\circ^\epsilon+h_\circ^{\epsilon}\,\sharp\,^{\epsilon}
\mathfrak{z}^{\epsilon}
+\mathfrak{z}^{\epsilon}\,\sharp\,^{
\epsilon}h_\circ^{\epsilon}\,\sharp\,^{\epsilon}\mathfrak{y}^{\epsilon}\big)  \phi^\epsilon_{\beta}\right\rangle_{\mathcal{H
}} \,+\, \kappa\epsilon \, \mathcal O(|\alpha-\beta||^{-\infty})\,,
\eeq
and 
\beq\label{S-2-magn-matrix}
\Sigma_2(\alpha,\beta)
\,=\,\widetilde{\Lambda}^{\epsilon,\kappa}(\alpha,\beta)
\left\langle\phi^\epsilon_{
\alpha}\,,\,\mathfrak{Op}^{\epsilon}\big(\mathfrak{y}^{\epsilon}\,\sharp\,^{\epsilon}h_\bullet^{\epsilon}
\,\sharp\,^{\epsilon}r_{\epsilon}\,\sharp\,^{\epsilon}
h_\bullet^{\epsilon}\,\sharp\,^{\epsilon}\mathfrak{y}^{\epsilon}\big)\phi^\epsilon_{\beta}\right\rangle_{\mathcal{H
}} +  \kappa \epsilon \,  \mathcal O(|\alpha-\beta|)^{-\infty}\,.
\eeq
This gives us \eqref{eq:7.2.1.a} in the proposition.

From Proposition \ref{Gamma-inv} it follows that the symbol
\beq
\mathfrak{h}^\epsilon:=h_\circ^\epsilon+h_\circ^{\epsilon}\,\sharp\,^{\epsilon}
\mathfrak{z}^{\epsilon}
+\mathfrak{z}^{\epsilon}\,\sharp\,^{
\epsilon}h_\circ^{\epsilon}\,\sharp\,^{\epsilon}\mathfrak{y}^{\epsilon}+\mathfrak{y}^{\epsilon}\,\sharp\,^{\epsilon}h_\bullet^{\epsilon}
\,\sharp\,^{\epsilon}r_{\epsilon}\,\sharp\,^{\epsilon}
h_\bullet^{\epsilon}\,\sharp\,^{\epsilon}\mathfrak{y}^{\epsilon}
\eeq
is $\Gamma$-periodic in the $\X$-variable. Thus, using the results in
Proposition \ref{Zak-transl},  $\mathfrak{Op}^\epsilon(\mathfrak{h}^\epsilon)$
commutes with
all the operators
$\{\Lambda^\epsilon(\gamma,\cdot)\tau_{\gamma}\}_{\gamma\in\Gamma}$ and we
get \eqref{eq:7.2.1.b}, achieving the proof of the proposition.
\end{proof}

\subsubsection {Control of $\Sigma_1 (\alpha,\beta)$.}\label{SS-sigma-1}

The main result here is the following lemma: 
\begin{lemma}\label{Concl-1}
We have
\begin{align*}
\Sigma_1(\alpha,\beta)&=\widetilde{\Lambda}^{\epsilon,\kappa}(\alpha,\beta)
\left\langle\phi^\epsilon_{
\alpha}\,,\,\left[H^{\epsilon}+\mathfrak{Op}^{\epsilon}\big(h_\circ^{\epsilon,
\kappa}\,\sharp^{\epsilon,\kappa}\, 
\mathfrak{z}^{\epsilon,\kappa}
\big)+\mathfrak{Op}^{\epsilon}\big(\mathfrak{z}^{\epsilon,\kappa}\,\sharp\,^{
\epsilon,\kappa}h_\circ^{\epsilon,\kappa}\,\sharp^{\epsilon,\kappa}\, 
\mathfrak{y}^{\epsilon,\kappa}\big)\right]\phi^\epsilon_{\beta}\right\rangle_{\mathcal{H
}}\\
&+ \kappa\epsilon\, \mathscr{O}(|\alpha-\beta|^{-\infty}) \,.
\end{align*}
\end{lemma}

The rest of this paragraph is dedicated to the proof of the above statement. We start by considering  the scalar products appearing in the double series in  the expression of  $\Sigma_1(\alpha,\beta)$ in \eqref{magn-matrix-1}.

\beq\label{sc-pr-1}
\begin{array}{ll}
\mathfrak{H}^{\epsilon,\kappa}_{\alpha'\beta'} & :=
\left\langle\big(Y^{\epsilon,\kappa}\big)^{-1/2}\Lambda^{\epsilon,\kappa}(\cdot,
\alpha')\tau_{\alpha'}\psi^\epsilon_0\,,\,H^{\epsilon,\kappa}
\big(Y^{\epsilon,\kappa}\big)^{-1/2}\Lambda^{\epsilon,\kappa}(\cdot,\beta')\tau_
{-\beta'}\psi^\epsilon_0\right\rangle_{\mathcal{H}}\\ & 
\; =\left\langle\phi^\epsilon_{\alpha'}\,,\,\widetilde{\Lambda}^{\epsilon,\kappa}(
\alpha',\cdot)H^{\epsilon,\kappa}\widetilde{\Lambda}^{\epsilon,\kappa}(\cdot,
\beta')\phi^\epsilon_{\beta'}\right\rangle_{\mathcal{H}} + \mathcal S_1(\alpha',\beta') + \mathcal S_2(\alpha',\beta')\,,
\end{array}
\eeq
with
\beq \label{T-2}
\mathcal S_1(\alpha',\beta'):=\left\langle\phi^\epsilon_{\alpha'}\,,\,\widetilde{\Lambda}^{\epsilon,\kappa}(
\alpha',\cdot)H^{\epsilon,\kappa}\big[\big(Y^{\epsilon,\kappa}\big)^{-1/2}
-\bb1\big]\widetilde{\Lambda}^{\epsilon,\kappa}(\cdot,
\beta')\phi^\epsilon_{\beta'}\right\rangle_{\mathcal{H}}
\eeq
 and
 \beq \label{T-3}
\mathcal S_2(\alpha',\beta'):= \,\left\langle\phi^\epsilon_{\alpha'}\,,\,\widetilde{\Lambda}^{\epsilon,\kappa}
(\alpha',\cdot)\big[\big(Y^{\epsilon,\kappa}\big)^{-1/2}
-\bb1\big]H^{\epsilon,\kappa}\big(Y^{\epsilon,\kappa}\big)^{-1/2}\widetilde{
\Lambda}^{ \epsilon,\kappa}(\cdot,
\beta')\phi^\epsilon_{\beta'}\right\rangle_{\mathcal{H}}.
\eeq

\begin{lemma}\label{P-red-H}
Let $m\in\mathbb{N}$. There exists $C_m>0$ such that for all  $\alpha',\beta' \in \Gamma$ and $0<\epsilon,\kappa\leq 1$ we have
 
$$
\left|\left\langle\phi^\epsilon_{\alpha'},\widetilde{\Lambda}^{\epsilon,
\kappa}(
\alpha',\cdot)H^{\epsilon,\kappa}\widetilde{\Lambda}^{\epsilon,\kappa}(\cdot,
\beta')\phi^\epsilon_{\beta'}\right\rangle_{\mathcal{H}}-\widetilde{\Lambda}
^{\epsilon,\kappa}(\alpha',\beta')
\left\langle\phi^{\epsilon}_{\alpha'},\,H^{\epsilon}\phi^{
\epsilon}_{\beta'}\right\rangle_{\mathcal{H}}\right|\leq
C_m\, \kappa\epsilon<\alpha'-\beta'>^{-m}.
$$
\end{lemma}

\begin{proof}
We use \eqref{26-02-2} for the pair of vector potentials
$(A^{\epsilon,\kappa},A^\epsilon)$, in the form
$$
(-i\nabla-A^{\epsilon,\kappa}(
x))^2\widetilde{\Lambda}^{\epsilon,\kappa}(x,\tilde{\beta}) 
=\widetilde{\Lambda}^{\epsilon,\kappa}(x,\tilde{\beta})\left(
-i\nabla-A^\epsilon(x)\,+\,\kappa a_\epsilon(x,\tilde{\beta})\right)^2\,,
$$
 writing:
\beq
a_\epsilon(x,z)_j\:=\underset{k}{\sum}(x-\gamma)_k\int_0^1\epsilon
B_{jk}\big(\epsilon\gamma+s\epsilon(x-\gamma)\big)s\,ds\, \mbox{ for } j=1,2\,,
\eeq
and noticing that
\begin{equation}\label{est-aepsilon}
|a_\epsilon(x,\gamma)|\leq C\, \epsilon<x-\gamma>\,.
\end{equation}
Using \eqref{D-Omega} we have 
\begin{equation}\label{L-ek}
\widetilde{\Lambda}^{\epsilon,\kappa}(x,\tilde{\alpha})^{-1}\widetilde{\Lambda}^
{\epsilon,\kappa}(x,\tilde{\beta})\,=\,\widetilde{\Lambda}^{\epsilon,\kappa}
(\tilde{\alpha},\tilde{\beta})\,
\widetilde{\Omega}^{\epsilon,\kappa}(\tilde{\alpha},x,\tilde{\beta})\,,
\end{equation}
and we know from \eqref{sh-not-magn-2} and \eqref{Est-Omega}   that
$$ |\widetilde{\Omega}^{\epsilon,\kappa}(\tilde{\alpha},x,\tilde{\beta})-\bb1|\,
\leq C \, \kappa\epsilon\, |x-\tilde { \alpha}|\, |x-\tilde{\beta}|\,.
$$

The fast decay of the quasi Wannier function (Proposition \ref{free-ps-W}) allows us to finish the proof.
\end{proof}

\begin{lemma}\label{P-red-OpF}
Let us consider an operator of the form
$
{\pi}^{\epsilon,\kappa}\mathfrak{Op}^{\epsilon,\kappa}(F)
\pi^{\epsilon,\kappa}
$
with $F\in S^-_1(\Xi)$ (see \eqref{M-inv}). Then for
any $m\in\mathbb{N}$, there exists a seminorm $\nu_m$ such that,  $\forall \,
(\alpha',\beta') \in \Gamma\times \Gamma$ and $\forall(\epsilon,\kappa)\in[0,1]\times[0,1]$,
$$
\begin{array}{l}
\left|\left\langle\phi^\epsilon_{\alpha'}\,,
\,\widetilde{\Lambda}^{\epsilon,\kappa}(\alpha',\cdot){\pi}^{\epsilon,
\kappa}\mathfrak{Op}^{\epsilon,\kappa}(F)
\pi^{\epsilon,\kappa}\widetilde{\Lambda}^{\epsilon,\kappa}(\cdot,
\beta')\phi^\epsilon_{\beta'}\right\rangle_{\mathcal{H}}\,-\,
\widetilde{\Lambda}^{\epsilon,\kappa}(\alpha',\beta')\left\langle\phi^{\epsilon}
_{\alpha'},\,{\pi}^{
\epsilon}\mathfrak{Op}^{\epsilon}(F)
\pi^{\epsilon}\,
\phi^{\epsilon}_{\beta'}\right\rangle_{\mathcal{H}}\right| 
\\
\quad \qquad\qquad 
\leq\nu_m(F)\, \kappa\epsilon<\alpha'-\beta'>^{-m}\,.
\end{array}
$$
\end{lemma}

\begin{proof}
Using Proposition B.8 in \cite{CHP} we may conclude that for 
$F\in S^-_1(\Xi)$ and for any $k\in\mathbb{N}$ there exists a seminorm
$\nu_k:S^-_1(\Xi)\rightarrow\mathbb{R}_+$ such that
\beq
\underset{x\in\X}{\sup}\int_\X \,|x-y|^k|K_F(x,y)\,dy |\leq \nu_k(F) \, .
\eeq
Let us compute for $u\in\mathcal{H}$:
\beq
\begin{array}{l}
\Big[\widetilde{\Lambda}^{\epsilon,\kappa}
(\alpha',\cdot)
\mathfrak{Op}^{\epsilon,\kappa}(F)
\widetilde{\Lambda}^{\epsilon,\kappa}(\cdot,\beta')u\Big](x)\\ \quad  =
\Big[\widetilde{\Lambda}^{\epsilon,\kappa}
(\alpha',\cdot)\big[\Int(\widetilde{\Lambda}^{\epsilon,\kappa}
\Lambda^\epsilon
K_F)(\widetilde{\Lambda}^{\epsilon,\kappa}(\cdot,\beta')u)\big]\Big](x) \\ \quad 
=\widetilde{\Lambda}^{\epsilon,\kappa}(\alpha',\beta')\widetilde{
\Omega} ^{\epsilon,\kappa}(\alpha',x,\beta')
\int_\X dy\,\widetilde{\Omega}
^{\epsilon,\kappa}(x,y,\beta')\Lambda^\epsilon(x,y)K_F(x,y)u(y)\,,
\end{array}
\eeq 
and take into account the estimates (from \eqref{Est-Omega})
\beq
\left|\widetilde{\Omega} ^{\epsilon,\kappa}(\alpha',x,\beta')-1\right|\leq
C\, \kappa\epsilon\, |x-\alpha'|\,|x-\beta'|\,,
\eeq
and 
\beq
\left|\widetilde{\Omega} ^{\epsilon,\kappa}(x,x+z,\beta')-1\right|\leq
C\, \kappa\epsilon\, |z|\,|x-\beta'|\,.
\eeq
Thus we obtain
\beq
\begin{array}{l}
\left|\left\langle
\phi^{\epsilon}_{\alpha},\,\widetilde{\Lambda}^{\epsilon,\kappa}(\alpha,\cdot)\pi^{\epsilon,\kappa}
\mathfrak{Op}^{\epsilon,\kappa}(F)\pi^{\epsilon,\kappa}\,
\widetilde{\Lambda}^{\epsilon,\kappa}(\cdot,\beta)
\phi^{\epsilon}_{\beta}\right\rangle_{\mathcal{H}}-
\widetilde{\Lambda}^{\epsilon,\kappa}(\alpha,\beta)\left\langle\phi^{\epsilon}
_\alpha,\,{\pi}^{
\epsilon}\mathfrak{Op}^{\epsilon}(F)\pi^{\epsilon}\,
\phi^{\epsilon}_\beta\right\rangle_{\mathcal{H}}\right| \\
\qquad\qquad\qquad\qquad\qquad\qquad\qquad  \leq \, \kappa\epsilon\, \nu_m(F)<\alpha'-\beta'>^m\,,
\end{array}
\eeq 
for any $m\in\mathbb{N}$, with some $\nu_m:S^0_0(\Xi)\rightarrow\mathbb{R}_+$ a
seminorm  on $S^0_0(\Xi)$.
\end{proof}

\begin{lemma}\label{H-bd-Lemma}
For any
$m\in\mathbb{N}$ there exists $C_m>0$ such that
$$
\left|\mathfrak{H}^{\epsilon,\kappa}_{\alpha'\beta'}\right|\leq
C_m<\alpha'-\beta'>^{-m}.
$$
\end{lemma}
\begin{proof}
We have 
\beq
\begin{array}{ll}
\mathfrak{H}^{\epsilon,\kappa}_{\alpha'\beta'}& =
\left\langle\big(Y^{\epsilon,\kappa}\big)^{-1/2}\Lambda^{\epsilon,\kappa}
(\cdot ,
\alpha')\tau_{\alpha'}\psi^\epsilon_0\,,\,H^{\epsilon,\kappa}
\big(Y^{\epsilon,\kappa}\big)^{-1/2}\Lambda^{\epsilon,\kappa}(\cdot,\beta')\tau_
{-\beta'}\psi^\epsilon_0\right\rangle_{\mathcal{H}} \\ & 
=\left\langle\phi^{\epsilon,\kappa}_{\alpha'}\,,\,\mathfrak{Op}^{\epsilon,\kappa
}\big(\mathfrak{y}^{\epsilon,\kappa}\,\sharp^{\epsilon,\kappa}\, 
h_\circ\,\sharp\,^
{\epsilon,\kappa} \mathfrak{y}^{\epsilon,\kappa}\big)
\phi^{\epsilon,\kappa}_{\beta'}\right\rangle_{\mathcal{H}}.
\end{array}
\eeq
We notice that
$h_\circ\,\sharp^{\epsilon,\kappa}\, 
\mathfrak{y}^{\epsilon,\kappa}\in
S^{-\infty}(\Xi)$ so that using once again Proposition B.8 from \cite{CHP} as
in the begining of the proof of Lemma \ref{P-red-OpF} and the rapid
decay of the magnetic distorted Wannier functions we conclude that for any
$m\in\mathbb{N}$ there exist $C_m>0$, $C'_m>0$ and a seminorm $\nu_m$ on $S^{-\infty}(\Xi)$, such that
$$
\begin{array}{l}
\left|\mathfrak{H}^{\epsilon,\kappa}_{\alpha'\beta'}\right| \\ 
\quad 
\leq C'_m\nu_m
\big(\mathfrak{y}^{\epsilon,\kappa}\,\sharp^{\epsilon,\kappa}\, 
h_\circ\,\sharp\,^{
\epsilon,\kappa}\mathfrak{y}^{\epsilon,\kappa}
\big)\int_\X\int_\X dx\,dy<x-\alpha'>^{-m-3}<x-y>^{-m}<y-\beta'>^{-m-3}\\
\quad \leq
 C_m<\alpha'-\beta'>^{-m}\,.
 \end{array}
$$
\end{proof}

\begin{remark}\label{d-series-coeff}
Let us notice that
\beq
\underset{\alpha'\in\Gamma}{
\sum}\underset{\beta'\in\Gamma}{\sum}\overline{\mathbb{F}^{\epsilon,\kappa}_{
\alpha'\alpha}}
\mathbb{F}^{\epsilon,\kappa}_{\beta'\beta}\mathfrak{H}^{\epsilon,\kappa}_{
\alpha',\beta'}=
\mathfrak{H}^{\epsilon,\kappa}_{\alpha,\beta}+\underset{\alpha'\in\Gamma}{
\sum}\big(\overline{\mathbb{F}^{\epsilon,\kappa}}-\bb1\big)_{\alpha'\alpha}
\mathfrak{H}^{\epsilon,\kappa}_{\alpha',\beta}
+\underset{\alpha'\in\Gamma}{
\sum}\underset{\beta'\in\Gamma}{\sum}\overline{\mathbb{F}^{\epsilon,\kappa}_{
\alpha'\alpha}}\big(
\mathbb{F}^{\epsilon,\kappa}-\bb1\big)_{\beta'\beta}\mathfrak{H}^{\epsilon,
\kappa}_{\alpha',\beta'}
\eeq
and due to Proposition \ref{Prop-m-q-W} (point 3) we deduce that for any $(\epsilon,\kappa)\in[0,\epsilon_0]\times[0,1]$ for some small enough $\epsilon_0>0$ and for any $(\alpha,\beta)\in\Gamma\times\Gamma$ we have that for any $m\in\mathbb{N}$ there exists some $C_m>0$ such that
\beq
\left|\underset{\alpha'\in\Gamma}{
\sum}\underset{\beta'\in\Gamma}{\sum}\overline{\mathbb{F}^{\epsilon,\kappa}_{
\alpha'\alpha}}
\mathbb{F}^{\epsilon,\kappa}_{\beta'\beta}\mathfrak{H}^{\epsilon,\kappa}_{
\alpha',\beta'}\,-\,\mathfrak{H}^{\epsilon,\kappa}_{\alpha,\beta}\right|\ \leq\ C_m\, \kappa\epsilon<\alpha-\beta>^{-m}.
\eeq
\end{remark}
Coming to the last two terms 
\eqref{T-2} and \eqref{T-3}   we notice that they may be written as:
\beq
\mathcal S_1
=\left\langle\phi^\epsilon_{\alpha'}\,,\,\widetilde{\Lambda}^{\epsilon,\kappa
}(\alpha',\cdot){\pi}^{\epsilon,\kappa}
\mathfrak{Op}^{\epsilon,\kappa}\big(h_\circ^{\epsilon,\kappa}\,\sharp\,^{
\epsilon,\kappa}\mathfrak{z}^{\epsilon,\kappa}
\big){\pi}^{\epsilon,\kappa}
\widetilde{\Lambda}^{\epsilon,\kappa}(\cdot,
\beta')\phi^\epsilon_{\beta'}\right\rangle_{\mathcal{H}},
\eeq
and 
\beq
\mathcal S_2
=\left\langle\phi^\epsilon_{\alpha'}\,,\,\widetilde{\Lambda}^{\epsilon,\kappa}
(\alpha',\cdot){\pi}^{\epsilon,\kappa}
\mathfrak{Op}^{\epsilon,\kappa}\big(\mathfrak{z}^{\epsilon,\kappa}\,\sharp\,^{
\epsilon,\kappa}h_\circ^{\epsilon,\kappa}\,\sharp^{\epsilon,\kappa}\, 
\mathfrak{y}^{\epsilon,\kappa}\big)
{\pi}^{\epsilon,\kappa}\widetilde{\Lambda}^{\epsilon,\kappa}(\cdot,
\beta')\phi^\epsilon_{\beta'}\right\rangle_{\mathcal{H}}.
\eeq
Thus
we notice that $h_\circ^{\epsilon,\kappa}\,\sharp\,^{
\epsilon,\kappa}\mathfrak{z}^{\epsilon,\kappa}$ and $
\mathfrak{z}^{\epsilon,\kappa}\,\sharp\,^{
\epsilon,\kappa}h_\circ^{\epsilon,\kappa}\,\sharp^{\epsilon,\kappa}\, 
\mathfrak{y}^{\epsilon,\kappa}$ are symbols of class $S^{-\infty}(\Xi)$ and we may use Lemma~\ref{P-red-OpF} in order to obtain that:
\begin{align*}
&\left|\left\langle\phi^\epsilon_{\alpha'}\,,\,
\widetilde{\Lambda}^{\epsilon,\kappa}(
\alpha',\cdot)H^{\epsilon,\kappa} \big[\big(Y^{\epsilon,\kappa}\big)^{-1/2}
-\bb1\big]\widetilde{\Lambda}^{\epsilon,\kappa}(\cdot,
\beta')\phi^\epsilon_{\beta'}\right\rangle_{\mathcal{H}} \right . \\
 & \qquad \qquad -\left . \widetilde{\Lambda}^{
\epsilon,\kappa}(\alpha',\beta')\left\langle\phi^\epsilon_{\alpha'}\,,\,
\mathfrak{Op}^{\epsilon}\big(h_\circ^{\epsilon,\kappa}\,\sharp\,^{
\epsilon,\kappa}\mathfrak{z}^{\epsilon,\kappa}
\big)\phi^\epsilon_{\beta'}\right\rangle_{\mathcal{H}}
\right| \leq   C_m\,\kappa\epsilon<\alpha'-\beta'>^{-m},
\end{align*}
and
\begin{align*}
&\left|\left\langle\phi^\epsilon_{\alpha'}\,,\,
\widetilde{\Lambda}^{\epsilon,
\kappa}
(\alpha',\cdot)\big[\big(Y^{\epsilon,\kappa}\big)^{-1/2}
-\bb1\big]H^{\epsilon,\kappa}\big(Y^{\epsilon,\kappa}\big)^{-1/2}\widetilde{
\Lambda}^{ \epsilon,\kappa}(\cdot,
\beta')\phi^\epsilon_{\beta'}\right\rangle_{\mathcal{H}}\right . \\
&\qquad \qquad - \left . \left\langle\phi^\epsilon_{\alpha'}\,,\,
\mathfrak{Op}^{\epsilon}\big(\mathfrak{z}^{\epsilon,\kappa}\,\sharp\,^{
\epsilon,\kappa}h_\circ^{\epsilon,\kappa}\,\sharp^{\epsilon,\kappa}\, 
\mathfrak{y}^{\epsilon,\kappa}\big)\phi^\epsilon_{\beta'}\right\rangle_{\mathcal{H}}
\right| \leq  C_m\, \kappa\epsilon<\alpha'-\beta'>^{-m}.
\end{align*}

\subsubsection {Control of $\Sigma_2(\alpha,\beta)$.}\label{SS-sigma-2}

Going back to the double series appearing  in  \eqref{magn-matrix-2} we write
\begin{align}\label{B-1}
&\mathfrak{R}^{\epsilon,\kappa}_{\alpha'\beta'}:= \nonumber\\
&=\left\langle\big(Y^{\epsilon,\kappa}\big)^{-1/2}\Lambda^{\epsilon,\kappa}(\cdot,
\alpha')\tau_{\alpha'}\psi^\epsilon_0\,,\,
{\pi}^{\epsilon,\kappa}H^{\epsilon,\kappa}
R^{\epsilon,\kappa}_\bot(0)
H^{\epsilon,\kappa}{\pi}^{\epsilon,\kappa}
\big(Y^{\epsilon,\kappa}\big)^{-1/2}\Lambda^{\epsilon,\kappa}(\cdot,\beta')\tau_
{-\beta'}\psi^\epsilon_0\right\rangle_{\mathcal{H}}\nonumber\\
&=\left\langle\phi^\epsilon_{\alpha'}\,,
\,\widetilde{\Lambda}^{\epsilon,\kappa}(\alpha',\cdot)\big(Y^{\epsilon,\kappa}
\big)^{-1/2}
{\pi}^{\epsilon,\kappa}H^{\epsilon,\kappa}
R^{\epsilon,\kappa}_\bot(0)
H^{\epsilon,\kappa}{\pi}^{\epsilon,\kappa}
\big(Y^{\epsilon,\kappa}\big)^{-1/2}\widetilde{\Lambda}^{\epsilon,\kappa}(\cdot,
\beta')\phi^\epsilon_{\beta'}\right\rangle_{\mathcal{H}}\nonumber\\
&=\left\langle\phi^\epsilon_{\alpha'}\,,
\,\widetilde{\Lambda}^{\epsilon,\kappa}(\alpha',\cdot)
\mathfrak{Op}^{\epsilon,\kappa}\big(
\mathfrak{y}^{\epsilon,\kappa}\,\sharp^{\epsilon,\kappa}\, 
h_\bullet^{\epsilon,
\kappa}\,\sharp^{\epsilon,\kappa}\, 
r_{\epsilon,\kappa}\,\sharp^{\epsilon,\kappa}\, 
h_\bullet^{\epsilon,\kappa}\,\sharp^{\epsilon,\kappa}\, 
\mathfrak{y}^{\epsilon,
\kappa}\big)
\widetilde{\Lambda}^{\epsilon,\kappa}(\cdot,
\beta')\phi^\epsilon_{\beta'}\right\rangle_{\mathcal{H}}\,,
\end{align}
and notice that
\begin{align*}
&\underset{\alpha'\in\Gamma}{
\sum}\underset{\beta'\in\Gamma}{\sum}\overline{\mathbb{F}}_{\alpha'\alpha}
\mathbb{F}_{\beta'\beta}\mathfrak{R}^{\epsilon,\kappa}_{\alpha'\beta'}\\
&=
\mathfrak{R}^{\epsilon,\kappa}_{\alpha,\beta}+\underset{\alpha'\in\Gamma}{
\sum}\big(\overline{\mathbb{F}^{\epsilon,\kappa}}-\bb1\big)_{\alpha'\alpha}
\mathfrak{R}^{\epsilon,\kappa}_{\alpha',\beta}
+\underset{\alpha'\in\Gamma}{
\sum}\underset{\beta'\in\Gamma}{\sum}\overline{\mathbb{F}^{\epsilon,\kappa}_{
\alpha'\alpha}}\big(
\mathbb{F}^{\epsilon,\kappa}-\bb1\big)_{\beta'\beta}\mathfrak{R}^{\epsilon,
\kappa}_{\alpha',\beta'}\,.
\end{align*}
Noticing further that
$h_\bullet^{\epsilon,\kappa}:=p^{\epsilon,\kappa}\,\sharp^{\epsilon,\kappa}\, 
h\,\sharp\,^{
\epsilon,\kappa}(1-p^{\epsilon,\kappa})\in S^{-\infty}(\Xi)$ and
using again Lemma~\ref{P-red-OpF} we obtain that for any
$m\in\mathbb{N}$ there exists $C_m>0$ such that
\beq\label{B-3}
\begin{array}{l}
\left|\mathfrak{R}^{\epsilon,\kappa}_{\alpha'\beta'}-\widetilde{\Lambda}^{
\epsilon,\kappa}(\alpha',\beta')\left\langle\phi^\epsilon_{
\alpha'},
\mathfrak{Op}^{\epsilon}\big(
\mathfrak{y}^{\epsilon,\kappa}\sharp^{\epsilon,\kappa}\, 
h_\bullet^{\epsilon,
\kappa}\sharp^{\epsilon,\kappa}
r_{\epsilon,\kappa}\,\sharp^{\epsilon,\kappa} 
h_\bullet^{\epsilon,\kappa}\,\sharp^{\epsilon,\kappa} 
\mathfrak{y}^{\epsilon,
\kappa}\big)\phi^\epsilon_{\beta'}\right\rangle_{\mathcal{H}}\right| \\
\quad \quad  \leq
C_m\kappa\epsilon<\alpha'-\beta'>^{-m}.
\end{array}
\eeq
Moreover the proof of Lemma \ref{H-bd-Lemma} applies in this situation also and
we obtain that for any
$m\in\mathbb{N}$ there exists $C_m>0$ such that
\beq\label{B-4}
\left|\mathfrak{R}^{\epsilon,\kappa}_{\alpha'\beta'}\right|\leq
C_m<\alpha'-\beta'>^{-m}.
\eeq
 Taking into account all these results (\eqref{B-1}-\eqref{B-4}) and point 3 of Proposition  \ref{Prop-m-q-W}
we obtain the following result.

\begin{lemma}\label{Concl-2} 
We have
$$
\Sigma_2 (\alpha,\beta)
=\widetilde{\Lambda}^{\epsilon,\kappa}(\alpha,\beta)
\left\langle\phi^\epsilon_{
\alpha}\,,\,\mathfrak{Op}^{\epsilon}\big(
\mathfrak{y}^{\epsilon,\kappa}\,\sharp^{\epsilon,\kappa}\, 
h_\bullet^{\epsilon,
\kappa}\,\sharp^{\epsilon,\kappa}\, 
r_{\epsilon,\kappa}\,\sharp^{\epsilon,\kappa}\, 
h_\bullet^{\epsilon,\kappa}\,\sharp^{\epsilon,\kappa}\, 
\mathfrak{y}^{\epsilon,
\kappa}\big)\phi^\epsilon_{\beta}\right\rangle_{\mathcal{H
}}+ \kappa \epsilon \, \mathscr{O}(|\alpha-\beta|^{-\infty})\,.
$$
\end{lemma}

\subsection{The modified energy band and proof of Theorem \ref{mainTh} {\rm(i)}}

Having in mind the result of Proposition \ref{prelim-res} let us introduce 
\beq\label{final-4}
k^\epsilon:\Gamma\rightarrow\mathbb{C},\ k^\epsilon(\gamma):=
\left\langle\phi^{\epsilon}_\alpha\,,\,\mathfrak{Op}
^\epsilon(\mathfrak{h}^\epsilon)
\phi^{\epsilon}_\beta
\right\rangle_{\mathcal{H}};\quad
\lambdabar^\epsilon(\theta):=\underset{\gamma\in\Gamma}{\sum}e^{-i<\theta,
\gamma>}k^\epsilon(\gamma).
\eeq
Proceeding as in \cite{CP-2} and \cite{CHP} we define the \textit{magnetic matrix} acting in $\ell^2(\Gamma)$:
\beq\label{magn-matrix}
\mathscr{M}^{\epsilon,\kappa}(\alpha,\beta):=\left\langle\phi^{\epsilon,\kappa}_\alpha\,,\,\mathfrak{Op}
^\epsilon(\mathfrak{h}^\epsilon)
\phi^{\epsilon,\kappa}_\beta
\right\rangle_{\mathcal{H}}=\Lambda^{\epsilon,\kappa}(\alpha,\beta)k^\epsilon(\alpha-\beta)\,.
\eeq

Let us consider the smooth function $\lambdabar^\epsilon:\mathbb{T}_*\rightarrow\mathbb{R}$ as a periodic smooth function on $\Xi$ constant along the directions in $\X\times\{0\}$ and write its magnetic quantization as an integral operator
\beq
\mathfrak{Op}^{\epsilon,\kappa}(\lambdabar^\epsilon):=\Int\big(
\Lambda^{\epsilon,\kappa} K_{{\lambda\hspace{-5pt}\relbar}^\epsilon}\big)\,,
\eeq
with
\beq
K_{{\lambda\hspace{-5pt}\relbar}^\epsilon}(x,y)=(2\pi)^{-2}\int_{\X^*}e^{i<\xi,x-y>}\lambdabar^\epsilon(\xi)d\xi=\underset{\gamma\in\Gamma}{\sum}k^\epsilon(\gamma)\delta_0(x-y-\gamma)\,.
\eeq
Let us define the following unitary map
\beq
\mathscr{W}_\Gamma:L^2(\X)\rightarrow\ell^2(\Gamma)\otimes L^2(E),\ \big(\mathscr{W}_\Gamma F\big)(\alpha,\underline{x}):=F(\alpha+\underline{x})
\eeq
 and compute the integral kernel $\mathfrak{K}$ of $\mathscr{W}_\Gamma\left(\mathfrak{Op}^{\epsilon,\kappa}(\lambdabar^\epsilon)\right) \mathscr{W}_\Gamma^{-1}$ in this representation
\beq
\mathfrak{K}(\alpha+\underline{x},\beta + \underline{y})=\Lambda^{\epsilon,\kappa}\big(\alpha+\underline{x},\beta+\underline{x}\big)k^\epsilon(\alpha-\beta)\delta_0(\underline{x}-
\underline{y})\,.
\eeq

In order to compare it with \eqref{magn-matrix} we shall consider two unitary gauge transformations:
\begin{align*}
\mathscr{U}_\epsilon:\ell^2(\Gamma)\otimes L^2(E)\rightarrow\ell^2(\Gamma)\otimes L^2(E),\quad\big(\mathscr{U}_\epsilon\Phi\big)(\alpha,\underline{x}):=\Lambda^\epsilon(\underline{x},\alpha)\Phi(\alpha,\underline{x})
\end{align*}
\begin{align*}
\mathscr{U}_{\epsilon,\kappa}:\ell^2(\Gamma)\otimes L^2(E)\rightarrow\ell^2(\Gamma)\otimes L^2(E),\quad\big(\mathscr{U}_{\epsilon,\kappa}\Phi\big)(\alpha,\underline{x}):=\widetilde{\Lambda}^{\epsilon,\kappa}(\alpha,\alpha+\underline{x})\Phi(\alpha,\underline{x}).
\end{align*}
We notice that, using also \eqref{D-Omega}, the kernel $\widetilde{ \mathfrak{K}}$ of $\big(\mathscr{U}_{\epsilon,\kappa}\mathscr{U}_\epsilon\mathscr{W}_\Gamma\big)\left(\mathfrak{Op}^{\epsilon,\kappa}(\lambdabar^\epsilon)\right) \big(\mathscr{U}_{\epsilon,\kappa}\mathscr{U}_\epsilon\mathscr{W}_\Gamma\big)^{-1}$ is given by:
\begin{align*}
&\widetilde{\mathfrak{K}}\big((\alpha,\underline{x});(\beta,\underline{y})\big)\\ &\qquad =
\widetilde{\Lambda}^{\epsilon,\kappa}(\alpha,\alpha+\underline{x})\Lambda^\epsilon(\underline{x},\alpha)\Lambda^{\epsilon,\kappa}(\alpha+\underline{x},\beta+\underline{x})\Lambda^\epsilon(\beta,\underline{x})\widetilde{\Lambda}^{\epsilon,\kappa}(\beta+\underline{x},\beta)k^\epsilon(\alpha-\beta)\delta_0(\underline{x}-
\underline{y})\\ 
&\qquad
=\Lambda^{\epsilon,\kappa}(\alpha,\beta)\widetilde{\Omega}^{\epsilon,\kappa}(\alpha,\alpha+\underline{x},\beta+\underline{x})\widetilde{\Omega}^{\epsilon,\kappa}(\alpha,\beta+\underline{x},\beta)k^\epsilon(\alpha-\beta)\delta_0(\underline{x}-
\underline{y})\,.
\end{align*}

\begin{proposition}\label{magn-matrix-ke}
The Hausdorff distance between the spectra of the operator $\mathfrak{Op}^{\epsilon,\kappa}(\lambdabar^\epsilon)$ in $\mathcal{L}\big(\mathcal{H}\big)$ and  the hermitian operator associated with the matrix $\mathscr{M}^{\epsilon,\kappa}(\alpha,\beta)$ in the orthonormal basis $\{\phi^{\epsilon,\kappa}_\gamma\}_{\gamma\in\Gamma}$ of $\pi^{\epsilon,\kappa}\mathcal{H}$ is of order $\kappa\epsilon$.
\end{proposition}

\begin{proof}
If we denote by
$x:=\alpha+\underline{x}$ and $x^\prime:=\beta+\underline{x}$ and use \eqref{Est-Omega} we obtain that
\beq
\left|1-\Omega^{\epsilon,\kappa}(\alpha,x,x^\prime)\right|\leq C\,\kappa\epsilon\, |\alpha-\beta|\,|\underline{x}|\leq C_1\, |\alpha-\beta|\, \kappa\epsilon,
\eeq
and
\beq
\left|1-\Omega^{\epsilon,\kappa}(\alpha,x^\prime,\beta)\right|\leq   C\kappa\epsilon|\alpha-\beta|\,|\underline{x}|\leq C_2\, |\alpha-\beta|\, \kappa\epsilon.
\eeq

Taking into account the rapid decay of $k^\epsilon(\gamma)$ for $|\gamma|\rightarrow\infty$ and considering the canonical orthonormal basis $\{\mathfrak{e}_\gamma\}_{\gamma\in\Gamma}$ of $\ell^2(\Gamma)$ defined by $\mathfrak{e}_\gamma(\alpha):=\delta_{\gamma,\alpha}$ we obtain
$$
\big \|\big(\mathscr{U}_{\epsilon,\kappa}\mathscr{U}_\epsilon\mathscr{W}_\Gamma\big)\left(\mathfrak
{Op}^{\epsilon,\kappa}(\lambdabar^\epsilon)\right) \big(\mathscr{U}_{\epsilon,\kappa}\mathscr{U}_\epsilon\mathscr{W}_\Gamma\big)^{-1}-
\hspace{-8pt}\underset{(\alpha,\beta)\in\Gamma\times\Gamma}{\sum}\mathscr{M}^{\epsilon,\kappa}(\alpha,\beta)\big(|\mathfrak{e}_\alpha\rangle\langle\mathfrak{e}_\beta|\otimes\bb1_{L^2(E)}\big)\big \|
\leq\ C\kappa\epsilon.$$ 
On the other hand it is evident that the two operators:
$$
\underset{(\alpha,\beta)\in\Gamma\times\Gamma}{\sum}\mathscr{M}^{\epsilon,\kappa}(\alpha,\beta)\big(|\mathfrak{e}_\alpha\rangle\langle\mathfrak{e}_\beta|\otimes\bb1_{L^2(E)}\big);\qquad
\widetilde{H}^{\epsilon,\kappa}=\underset{(\alpha,\beta)\in\Gamma\times\Gamma}{\sum}\mathscr{M}^{\epsilon,\kappa}(\alpha,\beta)|\phi^{\epsilon,\kappa}_\alpha\rangle\langle\phi^{\epsilon,\kappa}_\beta|
$$
have the same spectrum.
\end{proof}

Putting together the Propositions \ref{prelim-res}, \ref{magn-matrix-ke} and Proposition 3.19 in \cite{CHP} we obtain the third conclusion of our Theorem \ref{mainTh}.

\subsection{Behavior of the modified energy band function in the chosen energy window and proof of Theorem \ref{mainTh} {\rm(ii)}}\label{SS-def-band}

We will only prove Theorem \ref{mainTh} {\rm(ii)} in the case when $m=0$; the general case is similar due to the rapid decay of the quasi Wannier function $\psi_0$. 
\begin{proposition}\label{mod-Bloch-func}
For $b\in(0,\widetilde{b})$ with $\tilde b$ as in Lemma \ref{newlemma}  there exists $\epsilon_0>0$ and
$C>0$ such that, for any $\theta\in\Sigma_b\,$ and  any $\epsilon \in [0,\epsilon_0]$,
$$
\left|\lambdabar^\epsilon(\theta)-\lambda_{0}(\theta)\right|\leq 
C\,\epsilon.
$$
\end{proposition}
\begin{proof}
Let us also define the smooth function
\beq\label{final-5}
\lambdabar_\circ^\epsilon(\theta):=\underset
{\gamma\in\Gamma}{\sum}e^{-i<\theta,
\gamma>}\left\langle\phi^\epsilon_{\gamma}\,,\,
\mathfrak{Op}^\epsilon(h^\epsilon_\circ)\phi^\epsilon_{0}\right
\rangle_{\mathcal{H}}
\eeq
 and estimate the above difference  by a two step procedure:
\beq
\lambdabar^\epsilon(\theta)-\lambda_{0}(\theta)\ =\ \big(
\lambdabar^\epsilon(\theta)-\lambdabar_\circ^\epsilon(\theta)
\big)\,+\,\big(\lambdabar_\circ^\epsilon(\theta) 
-\lambda_{0}(\theta)\big).
\eeq

\noindent\textbf{Step 1.}
We begin by computing
$$
\lambdabar_\circ^\epsilon(\theta)-\lambda_{0}(\theta)\,=\,
\underset{\gamma\in\Gamma}{\sum}e^{-i<\theta,
\gamma>}\left\langle\phi^\epsilon_{\gamma}\,,\,
\mathfrak{Op}^\epsilon(h^\epsilon_\circ)\phi^\epsilon_{0}\right
\rangle_{\mathcal{H}}
-\lambda_{0}(\theta).
$$
Using Proposition \ref{Zak-transl} we obtain
$$ 
\left\langle\phi^\epsilon_{\gamma}\,,\,
\mathfrak{Op}^\epsilon(h^\epsilon_\circ)
\phi^\epsilon_{0}\right\rangle_{\mathcal{H}}  =
\left\langle\phi^\epsilon_{\gamma}\,,\,
\mathfrak{Op}^\epsilon(h)
\phi^\epsilon_{0}\right\rangle_{\mathcal{H}}\\ 
 =
\Lambda^\epsilon(\alpha,\beta)
\left\langle
\phi^\epsilon_{\alpha-\beta}\,,\,\mathfrak{Op}^
\epsilon(h)
\phi^\epsilon_{0}\right\rangle_{\mathcal{H}}\,.
$$
Taking also into account that we have fixed the gauge for $A_0$ as in \eqref{A0}, we conclude that
\beq\label{1-main}
\begin{array}{ll}
\left\langle\phi^\epsilon_{\gamma}\,,\,
\mathfrak{Op}^\epsilon(h)\phi^\epsilon_{0}\right
\rangle_{\mathcal{H}}& =
\left\langle\psi_{\gamma}\,,\,
H^0\psi_{0}\right\rangle_{\mathcal{H}}\\ & \quad 
+\, \epsilon \, \left[\left\langle\phi^\epsilon_{\gamma}\,,\,
\big(D\cdot A_0(\cdot)+A_0(\cdot)\cdot D+\epsilon
A_0(\cdot
)^2\big)\psi^\epsilon_{0}\right\rangle_{\mathcal{H}}
\right]\\ &\quad 
+\left[\left\langle\big(\tau_{\gamma}(\psi^\epsilon_
{0}-\psi_{0}) \big)\,,\,
H^0\psi^\epsilon_{0}\right\rangle_{\mathcal{H}}+
\left\langle\tau_{\gamma}\psi^\epsilon_{0}\,,\,
H^0\big(\psi^\epsilon_{0}-\psi_0\big)\right\rangle
_{\mathcal{H}}\right]
\\ & \quad 
+\left\langle\tau_{\gamma}\psi^\epsilon_{0}\,,\,\big(\Omega^\epsilon(\gamma,0,
\cdot)-1\big)
\big((D+\epsilon
a_0(\cdot,0))^2+V(\cdot)\big)\psi^\epsilon_{0}\right\rangle_{\mathcal{H}}.
\end{array}
\eeq
We recall from
Definition \ref{m-q-W-def} that, for any $\gamma\in\Gamma$,
$$
\phi^\epsilon_{\gamma}=\Lambda^\epsilon(\cdot,\gamma)\tau_{\gamma}
\underset{\beta\in\Gamma}{\sum}
\Omega^\epsilon(\beta,0,\cdot)\mathbf{F}
^\epsilon(\beta)\tau_{\beta}\psi_0\,.
$$
 Taking into account 
the rapid decay of $\mathbf{F}$ (see Proposition \ref{Prop-m-q-W}) and of  $\psi_0$,  we conclude that for any $m\in\mathbb{N}$ there exists $C_m>0$
such that for any $\epsilon\in[0,\epsilon_0]$:
$$
\left|\left\langle\phi^\epsilon_{\gamma}\,,\,
\big(D\cdot A_0(\cdot)+A_0(\cdot)\cdot D+\epsilon
A_0(\cdot)^2\big)\psi^\epsilon_{0}\right\rangle_
{\mathcal{H}}\right|\leq
C_m<\gamma>^{-m},\quad\forall\gamma\in\Gamma,
$$
and 
$$
\left|\left\langle\tau_{\gamma}\psi^\epsilon_{0}\,,\,\big(\Omega^\epsilon(\gamma,0,
\cdot)-1\big)
\big((D+\epsilon
A_0(\cdot
))^2+V(\cdot)\big)\psi^\epsilon_{0}\right\rangle_
{\mathcal{H}}\right|
\leq C_m\,\epsilon<\gamma>^{-m},\quad\forall\gamma\in\Gamma.
$$
Using Definition \ref{m-q-W-def} and point (2) in Proposition \ref{Prop-m-q-W} we get also that for any $m\in\mathbb{N}$
there exists $C_m>0$ such that for any
$\epsilon\in[0,\epsilon_0]$:
$$
\left|\left\langle\big(\tau_{\gamma}(\psi^\epsilon_
{0}-\psi_{0}) \big)\,,\,
H^0\psi^\epsilon_{0}\right\rangle_{\mathcal{H}}+
\left\langle\tau_{\gamma}\psi^\epsilon_{0}\,,\,
H^0\big(\psi^\epsilon_{0}-\psi_0\big)\right\rangle_
\mathcal{H}\right|
\leq C_m\, \epsilon<\gamma>^{-m},\quad\forall\gamma\in\Gamma.
$$
We have thus obtained:
$$
\lambdabar_\circ^\epsilon(\theta)-\lambda_0(\theta)
\,=\,\underset{\gamma\in\Gamma}{\sum}e^{-i<\theta,
\gamma>}\left\langle\psi_{\gamma}\,,\,
H^0\psi_{0}\right\rangle_{\mathcal{H}}
-\lambda_{0}(\theta)\,+\,\mathcal O (\epsilon)\,.
$$
Let us now use the Bloch-Floquet representation and the properties of $\hat{\psi}_0$ as constructed in Section~\ref{S_q-W} in order to write:
$$
\begin{array}{ll}
\left\langle\psi_{\gamma}\,,\,
H^0\psi_{0}\right\rangle_{\mathcal{H} } &=\int_{\mathbb{T}_*}d\omega\,
e^{i<\omega,\gamma>}\left\langle\hat{\psi}_{0}(\omega)\,,\left(\underset{n\geq0}
{\sum}\lambda_n(\omega)|\hat{\phi}_n(\omega)\rangle\langle\hat{\phi}
_n(\omega)|\right)\hat{\psi}_{0}(\omega)\right\rangle_{\mathscr{F}_\omega}
\\&
=\underset{n\in\mathbb{N}}{\sum}\int_{\mathbb{T}_*}
d\omega\,e^{i<\omega,\gamma>}\lambda_n(\omega)\left|\left\langle\hat{
\psi}_{0}(\omega)\,,\hat{\phi}_n(\omega)\right\rangle_{
\mathscr{F}_\omega}\right|^2\,,
\end{array}
$$
and we notice that for
$\omega\in\Sigma_b$ we have $\hat{
\psi}_{0}(\omega)=\hat{
\phi}_{0}(\omega)$.
 Finally we can write:
$$
\underset{\gamma\in\Gamma}{\sum}e^{-i<\theta,
\gamma>}\left\langle\psi_{\gamma}\,,\,
H^0\psi_{0}\right\rangle_{\mathcal{H}}=\underset{n\in
\mathbb{N}}{\sum} 
\lambda_n(\theta)\left|\left\langle\hat{
\psi}_{0}(\theta)\,,\hat{\phi}_n(\theta)\right\rangle_{
\mathscr{F}_\omega}\right|^2
$$
and thus 
$$\underset{\gamma\in\Gamma}{\sum}e^{-i<\theta,
\gamma>}\left\langle\psi_{\gamma}\,,\,
H^0\psi_{0}\right\rangle_{\mathcal{H}}=\lambda_{0}(\theta)\mbox{ for } \theta\in\Sigma_b\,.$$

\noindent\textbf{Step 2.}
We have to study the difference
$$
\lambdabar^\epsilon(\theta)-\lambdabar^\epsilon_\circ(\theta)
=\underset{\gamma\in\Gamma}{\sum}e^{-i<\theta,\gamma>}\Big(\langle\phi^\epsilon_\gamma,\mathfrak{Op}^\epsilon(\mathfrak{h}^\epsilon)\phi^\epsilon_0
\rangle_{\mathcal{H}}\ -\ \langle\phi^\epsilon_\gamma,\mathfrak{Op}^\epsilon(h^\epsilon_\circ)\phi^\epsilon_0
\rangle_{\mathcal{H}}\Big).
$$
Thus let us analyse the symbol: 
$$
\mathfrak{h}^\epsilon-h_\circ^\epsilon=\big(h_\circ^{\epsilon}\,\sharp\,^{\epsilon}
\mathfrak{z}^{\epsilon}+\mathfrak{z}^{\epsilon}\,\sharp\,^{
\epsilon}h_\circ^{\epsilon}\big)
+\mathfrak{z}^{\epsilon}\,\sharp\,^{
\epsilon}h_\circ^{\epsilon}\,\sharp\,^{\epsilon}\mathfrak{z}
^{\epsilon}+\mathfrak{y}^{\epsilon}\,\sharp\,^{\epsilon}h_\bot^{\epsilon}
\,\sharp\,^{\epsilon}r_{\epsilon}\,\sharp\,^{\epsilon}
h_\bot^{\epsilon}\,\sharp\,^{\epsilon}\mathfrak{y}^{\epsilon}\,,
$$
where $\mathfrak{y}^\epsilon$ and $\mathfrak{z}^\epsilon$ are defined in the last item in Definition \ref{rem-Hpi}. \\
Let $h_\circ\in S^{-\infty}(\Xi)$ be the symbol of the bounded operator
$$
\pi H^0\pi=\mathscr{U}_\Gamma^{-1}\left(\int_{\mathbb{T}_*}^
\oplus\hat{\pi}(\theta)\hat{H}^0(\theta)\hat{\pi}(\theta)\,d\theta\right)\mathscr{U}_\Gamma\,,
$$
$h_\bullet\in S^{-\infty}(\Xi)$ be the symbol of the operator $H_\bullet:=\pi H^0\pi^\bot\,$,  and $h^*_\bullet\in S^{-\infty}(\Xi)$ be the symbol of its adjoint. Moreover, using \eqref{l-bound-0} and then Proposition \ref{red-mPsD}, we denote by $R_\bot$ the inverse of $\pi^\bot H^0\pi^\bot$ as operator acting in $\pi^\bot\mathcal{H}$ and let $r\in S^{-m}_1(\Xi)$ be its symbol. 

Using Proposition \ref{Zak-transl} and Proposition \ref{Prop-m-q-W} and the fast decay of the modified Wannier function we can prove:
\beq\label{final-2}
\begin{array}{l}
\left\langle\phi^\epsilon_{\gamma}\,,\,\mathfrak{Op}^{\epsilon}\big((h_\circ^{\epsilon}\,\sharp\,^{\epsilon}
\mathfrak{z}^{\epsilon}+\mathfrak{z}^{\epsilon}\,\sharp\,^{
\epsilon}h_\circ^{\epsilon})+
\mathfrak{z}^{\epsilon}\,\sharp\,^{
\epsilon}h_\circ^{\epsilon}\,\sharp\,^{\epsilon}\mathfrak{z}
^{\epsilon}+\mathfrak{y}^{\epsilon}\,\sharp\,^{\epsilon}h_\bullet^{\epsilon}
\,\sharp\,^{\epsilon}r_{\epsilon}\,\sharp\,^{\epsilon}
(h_\bullet^{\epsilon})^*\,\sharp\,^{\epsilon}\mathfrak{y}^{\epsilon}
\big)\phi^\epsilon_0\right\rangle_{\mathcal{H}}\\
\qquad \qquad =\left\langle\tau_{\gamma}\psi_0
\,,\, Z \psi_0\right\rangle_{\mathcal{H}}
+\epsilon\, \mathscr{O}(<\gamma>^{-m})\,,
\end{array}
\eeq
with
\begin{equation}\label{eq:7.67}
Z:= \mathfrak{Op}\big((h_\circ\,\sharp\,
\mathfrak{z}+\mathfrak{z}\,\sharp\, h_\circ)+
\mathfrak{z}\,\sharp\, h_\circ\,\sharp\,\mathfrak{z}
+\mathfrak{y}\,\sharp\, h_\bullet
\,\sharp\, r\,\sharp\,
h_\bullet^*\,\sharp\,\mathfrak{y}
\big)\,.
\end{equation}
Here we also used that $\Lambda^\epsilon(\gamma,0)=1$ for any $\gamma\in\Gamma$.\\
The operator $Z$ has the following form:
\beq\label{final-II-0}
\begin{array}{ll}
Z &=H^0\big(Y^{-1/2}-\bb1\big)+\big(Y^{-1/2}-\bb1\big)H^0+
\big(Y^{-1/2}-\bb1\big)H^0\big(Y^{-1/2}-\bb1\big)\\ & \quad 
+Y^{-1/2}\pi H^0\pi^\bot R_\bot\pi^\bot H^0\pi Y^{-1/2}\,,
\end{array}
\eeq
and recall from \eqref{dress} that 
$$Y:=\pi H^0\pi^\bot R_\bot^2\pi^\bot H^0\pi\,.
$$
 All the operators appearing in \eqref{final-II-0}  have evidently $\Gamma$-periodic symbols and thus also $Z$ and $Y$.
More precisely we have the following direct integral decomposition:
$$
Y=\pi H^0\pi^\bot R_\bot^2\pi^\bot H^0\pi\equiv H_\bullet R_\bot^2H_\bullet=\mathscr{U}_\Gamma^{-1}\left(\int_{\mathbb{T}_*}^
\oplus d\theta\,
\hat{Y}(\theta)\right)\mathscr{U}_\Gamma,
$$
with
 $$
\hat{Y}(\theta):=\hat{H}_\bullet(\theta)\hat{R}_\bot
(\theta)^2\big[\hat{H}
_\bullet(\theta)\big]^*,
$$
 and $\hat{H}_\bullet(\theta)$ maps $\hat{\pi}^\bot(\theta)\mathscr{F}
_\theta$ into $\hat{\pi}(\theta)\mathscr{F}_\theta$ (see \eqref{hatpi} for the definition of $\hat \pi(\theta)$) and is defined by
$$ 
\hat{H}_\bullet(\theta)
=\hat \pi(\theta) 
\Big(\underset{n\in\mathbb{N}}{\sum}\lambda_n(\theta)|\hat{\phi}
_n(\theta)\rangle\langle\hat{\phi}_n(\theta)|\Big)\, \hat \pi(\theta)^\perp\,.
$$
Thus $$
Y=\pi H^0\pi^\bot R_\bot^2\pi^\bot H^0\pi=\mathscr{U}_\Gamma^{-1}\left(\int_{\mathbb{T}_*}^
\oplus d\theta
 \hat{Y}(\theta)
|\hat{\psi}_0(\theta)\rangle\langle\hat{\psi}_0(\theta)|\,,
\right)\mathscr{U}_\Gamma$$ 
with $ \hat{Y}\in C^\infty(\mathbb{T}_*;\mathbb{R}_+)$.\\
Taking into account the definitions and arguments in Section \ref{S_q-W}, 
we obtain that, for $\theta\in\Sigma_b\,$, 
$$
\hat{H}_\bullet(\theta)=|\hat{\phi}_0(\theta)\rangle\langle\hat{\phi}_0(\theta)|
\Big(\underset{n\in\mathbb{N}}{\sum}\lambda_n(\theta)|\hat{\phi}
_n(\theta)\rangle\langle\hat{\phi}_n(\theta)|\Big)\big(\bb1-
|\hat{\phi}_0(\theta)\rangle\langle\hat{\phi}_0(\theta)|\big)=0\,,
$$
so that we conclude that $ \hat{Y}(\theta)=0$ for $\theta\in \Sigma_b\,$. 

Now the operator $Z$ in \eqref{eq:7.67}  is also $\Gamma$-decomposable and 
$$
Z\ =\ \mathscr{U}_\Gamma^{-1}\left(\int^\oplus_
{\mathbb{T}_*}d\theta\, \hat{Z}(\theta)\right)\mathscr{U}_\Gamma ,
$$
where $ \hat{Z}\in C^\infty(\mathbb{T}_*)$ satisfies 
$
 \hat{Z}(\theta)=0$, $\forall\theta\in\Sigma_b$.
Finally, we obtain that
$$
\begin{array}{ll}
\lambdabar^\epsilon(\theta)-\lambdabar^\epsilon_\circ(\theta) &=\underset{\gamma\in\Gamma}{\sum}e^{-i<\theta,\gamma>}\langle\phi^\epsilon_\gamma,\mathfrak{Op}^\epsilon(\mathfrak{k}^\epsilon-h^\epsilon_\circ)\phi^\epsilon_0
\rangle_{\mathcal{H}}\\
&=\underset{\gamma\in\Gamma}{\sum}e^{-i<\theta,\gamma>}\langle\tau_{\gamma}\psi_0,Z\psi_0
\rangle_{\mathcal{H}}\,+\,\mathscr{O}(\epsilon)\\ & =\  \hat{Z}(\theta)\,+\,\mathscr{O}(\epsilon)=\mathscr{O}(\epsilon),\quad \forall\theta\in\Sigma_b,
\end{array}
$$
which ends the proof of the proposition and of Theorem \ref{mainTh} (ii). 
\end{proof}

\subsection{Proof of Corollaries \ref{C1} and \ref{C2}}\label{hcfin}
 
First, an application of Proposition \ref{magn-matrix-ke} and Proposition \ref{C-FS} with $\eta=\epsilon$ shows that the spectrum of $H^{\epsilon,\kappa}$ must have gaps of order $\epsilon$ in the interval $[0,N\epsilon]$, provided the same is true for $\mathfrak{Op}^{\epsilon,\kappa}(\widetilde{\lambdabar^\epsilon})$. 

Second, one has to perform the spectral analysis of $\mathfrak{Op}^{\epsilon,\kappa}(\widetilde{\lambdabar^\epsilon})$ applying the results in \cite{CHP} and prove the existence of gaps of order $\epsilon$ independent of $\kappa$ in its spectrum restricted to the interval $[0,N\epsilon]$, provided $\kappa$ and $\epsilon$ are small enough.

\begin{itemize}
\item  Assume $B^\Gamma=0$. Then we notice that the function $\lambdabar^\epsilon:\mathbb{T}_*\rightarrow\mathbb{R}$ has exactly the same properties as the function $\lambda^\epsilon:\mathbb{T}_*\rightarrow\mathbb{R}$ defined in Definition~3.18 in \cite{CHP} with the only difference that the estimate $|\partial^\alpha\lambda^\epsilon(\theta)-\partial^\alpha\lambda_0(\theta)|\leq C_m\epsilon$ which was valid for all $\theta\in\mathbb{T}_*$ (as stated in Proposition 4.1 in \cite{CHP}) is now only valid on $\Sigma_b\,$. Nevertheless, if we choose an upper bound $\delta_0$ for the cut-off parameter $\delta>0$ in Subsection4.3 (Paragraph 'Cut-off functions near the minimum') in \cite{CHP} such that $\{\theta\in E_*\mid|\theta|\leq\sqrt{2m_1^{-1}}\delta_0\}\subset\Sigma_b$, all the results of Section 4 of \cite{CHP} remain true for our function $\lambdabar^\epsilon:\mathbb{T}_*\rightarrow\mathbb{R}$ 
 and Corollary \ref{C1} follows.

\item Assume $B^\Gamma\ne0$. Then the situation is rather similar, but the function $\lambda_{0}:\mathbb{T}_*\rightarrow\mathbb{R}$ may no longer be symmetric around its local minimum $\theta_0$. In this case, a third order term may appear in the Taylor expansion (4.1) in \cite{CHP} and thus formula (4.3) in \cite{CHP} becomes
\beq
\lambda^\epsilon (\theta) - \lambda^\epsilon (\theta^\epsilon) =
 \,\underset{
1\leq j,k\leq2}{\sum} a_{jk}^\epsilon  (\theta_j-\theta^\epsilon_j)(\theta_k-\theta^\epsilon_k)  + \mathscr{O}(|\theta-\theta^\epsilon|^3) .
\eeq
As a consequence, the exponent $\mu>0$ relating the cut-off parameter $\delta>0$ with the intensity of the magnetic field $\epsilon>0$ through condition (4.20) in \cite{CHP} may vary only in the interval $(2,3)$ in $\mathbb{R}$. A 'symmetric choice', similar to the one in \cite{CHP} is thus $\mu=2.5$ and the only effect of this new choice is that in the second formula (4.28) in Proposition 4.5 in \cite{CHP} we now have the estimate $$ \|\mathfrak{Op}^{\epsilon,\kappa}
(\mathfrak{r}_{\delta,a})\|\,\leq\,C\;\epsilon^{1/5}\,.
$$
This ends the proof of Corollary \ref{C2}.
\end{itemize}

\vspace{0.2cm}

\noindent {\bf \large Acknowledgements}. Two of the authors (BH and RP) had an extended stay at the Department of Mathematical Sciences, Aalborg University, period during which this work had been finalized. The financial support is gratefully acknowledged.

\end{document}